\documentclass[journal]{IEEEtran}
\ifCLASSINFOpdf
\else
\fi
\usepackage{graphicx}
\usepackage{epstopdf}
\usepackage{epsfig}

\usepackage{cite}
\usepackage{amsmath,amssymb,amsfonts}
\usepackage{algorithmic}
\usepackage{textcomp}
\usepackage[latin1]{inputenc}
\usepackage{amsmath}
\usepackage{amsfonts}
\usepackage{amssymb}
\usepackage{makeidx}

\usepackage{mathrsfs}

\usepackage{lipsum}
\usepackage{flushend}
\usepackage{cuted}
\usepackage{mathtools}

\usepackage{CJK}
\usepackage{type1cm}
\usepackage{times}
\usepackage{footmisc}
\usepackage{cite}
\usepackage{dblfloatfix}

\usepackage{algorithm}
\usepackage{algorithmic}

\usepackage[latin1]{inputenc}

\usepackage{amsthm}

\newtheorem{theorem}{Theorem}

\newtheorem{lemma}{Lemma}

\newtheorem{problem}{Problem}
\newtheorem{proposition}{Proposition}
\newtheorem{remark}{Remark}

\DeclareMathOperator*{\ve}{{\rm vec}}
\DeclareMathOperator*{\veh}{\rm vech}
\DeclareMathOperator*{\Vex}{\rm Vex}

\begin{document}
	
	\title{Stochastic Event-based Sensor Schedules for Remote State Estimation in Cognitive Radio Sensor Networks}
	\author{Lingying Huang,
		Jiazheng Wang,
		Enoch Kung, 
		Yilin Mo, 
		Junfeng Wu and
		Ling Shi
		\thanks{This paper was accepted by IEEE Transaction on Automatic Control. It was published as an early access on July 7 2020 with the link: https://ieeexplore.ieee.org/document/9134874.}
		\thanks{L. Huang, J. Wang and L. Shi are with the Department of Electronic and Computer Engineering, Hong Kong University of Science and Technology, Clear Water Bay, Kowloon, Hong Kong (e-mail: lhuangaq@connect.ust.hk, jwangck@connect.ust.hk, eesling@ust.hk). The work by L. Huang, J. Wang  and L. Shi is supported by a Hong Kong RGC General Research Fund 16204218.}
		\thanks{E. Kung is the corresponding author with the Department of Mathematics, University College London, London,  Greater London, United Kingdom (email: e.kung@ucl.ac.uk). }
		\thanks{Y. Mo is with the Department of Automation and BNRist, Tsinghua University, Beijing, China (email: ylmo@tsinghua.edu.cn). The work by Y. Mo is supported by the National Key Research and Development Program of China under Grant 2018AAA0101601.}
		\thanks{J. Wu is with the College of Control Science and Engineering, Zhejiang University, Hangzhou, China (email: jfwu@zju.edu.cn).The work by J. Wu is supported by Natural Science Foundation of China under NSFC 61790571.}}
	
	\maketitle
	
	\begin{abstract}
		We consider the problem of communication allocation for remote state estimation in a cognitive radio sensor network~(CRSN). A sensor collects measurements of a physical plant, and transmits the data to a remote estimator as a secondary user (SU) in the shared network. The existence of the primal users (PUs) brings exogenous uncertainties into the transmission scheduling process, and how to design an event-based scheduling scheme considering these uncertainties has not been addressed in the literature. In this work, we start from the formulation of a discrete-time remote estimation process in the CRSN, and then analyze the hidden information contained in the absence of data transmission. In order to achieve a better tradeoff between estimation performance and communication consumption, we propose both open-loop and closed-loop schedules using the hidden information under a Bayesian setting. The open-loop schedule does not rely on any feedback signal but only works for stable plants. For unstable plants, a closed-loop schedule is designed based on feedback signals. The parameter design problems in both schedules are efficiently solved by convex programming. Numerical simulations are included to illustrate the theoretical results.
	\end{abstract}
	
	\begin{IEEEkeywords}
		Stochastic event-based schedule;  Cognitive radio sensor network; Minimum mean squared error; Branch-and-bound algorithm.
	\end{IEEEkeywords}
	
	\section{Introduction}
Recently,  cognitive ratio (CR) which \emph{dynamically} assigns the radio resources is applied in 5G Internet of things (IoT) applications\cite{hasegawa2014optimization}. 
CR, first proposed by Mitola et al.~\cite{mitola1999cognitive} in 1999, is a promising technology to cope with the spectrum scarcity problem. A CR sensor network~(CRSN) is a network of dispersed wireless sensor nodes embedded with cognitive radio capability	which enables them to dynamically access unused licensed spectrum bands for data transmission while	performing conventional wireless sensor nodes' tasks \cite{akan2009cognitive}. An example of that is shown in Fig. \ref{CRSN}. If the primary users (PUs), as the licensed user (mobile phone), vacate the spectrum, secondary users (SUs), e.g., the sink, equipped with CR devices can then  access the spectrum to transmit packets~\cite{kakalou2017cognitive}. Minimizing the communication rate of the SU while satisfying the estimation performance is worth studying in this shared network. 

Proper sensor scheduling, which is introduced to cope with limited transmissions, could improve estimation quality. The use of online information in event-based mechanisms to outperform off-line mechanisms ~\cite{yang2011deterministic,shi2011sensor,mo2014infinite}, in terms of estimation quality, has attracted increasing attention in recent years. Astrom and Bernhardsson~\cite{astrom2002comparison} first showed that an event-based approach outperforms a periodic approach (Riemann sampling) in a first-order stochastic system. 
The event-triggered mechanisms proposed by Xia et al.~\cite{xia2017networked} and Trimpe et al.~\cite{trimpe2014stability} require that the sensor has a computational capability to run a local Kalman filter and obtain a local state estimate. 
In realistic scenarios, however, the sensors may be primitive and have limited computational capability. Based on that condition, Wu et al.~\cite{wu2013event} derived a minimum mean squared error (MMSE) estimate on the remote estimator under a deterministic event-triggered scheduler. Since finding the exact MMSE estimate is intractable due to the computational complexity, an approximated estimator based on a Gaussian assumption is further derived. To preserve the Gaussian property, stochastic event-triggered sensor schedulers are proposed by Han et al.~\cite{han2015stochastic}. 

Different from traditional studies, in which the radio access network is \emph{statically} assigned, the existence of PUs introduces an exogenous uncertainty to the SU base scheduling scheme. There is a limited amount of works on optimizing the scheduling scheme of CRSNs. Deng et al.~\cite{deng2011energy} studied how to activate successively non-disjoint sensor groups to extend the network lifetime. Mabrouk et al.~\cite{mabrouk2014oticor} introduced opportunistic time slot assignment scheduling scheme to minimize the schedule length and maximize the throughput. All the above setups consider continuous-time measurements of the SU. Minimizing the transmission collision from a probabilistic point of view is important since it is very energy-consuming or even impossible to check the spectrum availability continuously. Moreover, the above studies neglect the information's importance? 

In this paper, we consider a discrete-time remote estimation process in a CRSN. Unlike previous studies, the SU can check the spectrum availability before each transmission. Moreover, we use an event-triggered mechanism to capture the information's importance. To the best of our knowledge, an event-based mechanism for remote state estimation has not been studied in this new but widely-used network structure.
	
	\begin{figure}[!t]
		\centering
		\includegraphics[width=0.9\columnwidth]{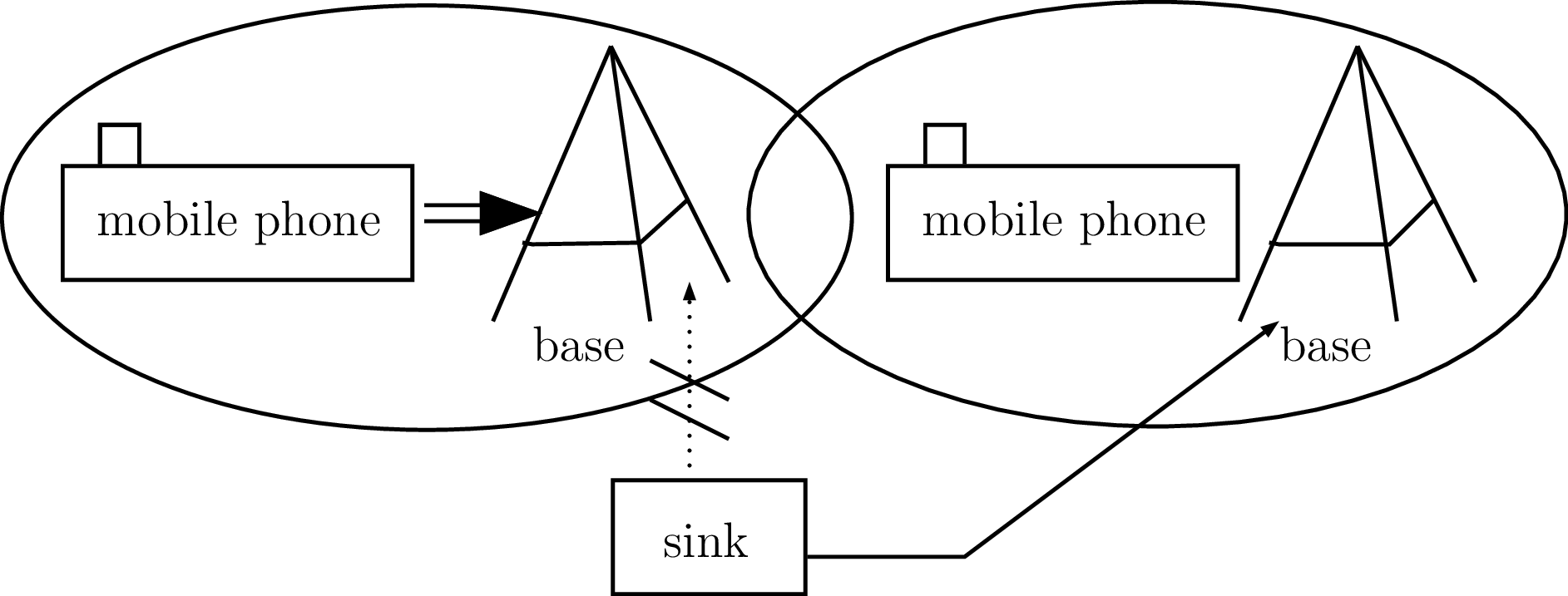}	
		\caption{Topology of a typical CRSN.(In this topology, mobile phone is a PU and it could get the access to the network whenever it has a packet to transmit. The sink, equipped with CR devices, is a SU base station and it could only access the spectrum when the PU vacates it.)}
		\label{CRSN}
	\end{figure}
	
The big challenge is that the exogenous uncertainty in the shared network in addition to the stochastic property of triggering law makes the uncertainties coupled. Kung et al.~\cite{kung2017nonexistence} showed that the Gaussian property cannot be preserved due to coupled uncertainty induced by packet drops. Xu et al.~\cite{xu2019remote} utilizes a Gaussian mixture model to obtain a closed-form MMSE estimator for the packet-dropping scenarios; however, the computational complexity grows exponentially. To cope with this challenge, we utilize the hidden information contained in the absence of transmission data in the CRSN to decouple those uncertainties. To be more specific, since the remote estimator can distinguish the source of the received packet, not triggering is inferred when no packet is received. Another problem is that, the exogenous uncertainty makes the error covariance random, then  deriving the error covariance bounds is highly non-trivial. The mean error covariance and its bounds under certain communication rate are first analyzed with a centralized base station collecting the different SUs' measurements within the channel coverage. This result can be used to analyze multiple SUs without a centralized base station as a future work, which is of great importance for 5G IoT. The main contributions of this work are summarized as follows.
\begin{enumerate}	
	\item  \textbf{Exogenous Uncertainty Model.}	 The novelty of the formulation is  taking into account the uncertain access to the network.   When $ \lambda =1$, Section $ \rm \uppercase\expandafter{\romannumeral3} $ recovers the work in~\cite{han2015stochastic}.
	\item \textbf{MMSE Estimator and Performance Bounds.} We derive the MMSE estimator for both the open-loop and closed-loop schedulers under this new model~(\textbf{Theorem~\ref{t1},~\ref{t3}}). Moreover, the (asymptotic) upper and lower bounds on the mean error  covariance are characterized~(\textbf{Lemma \ref{l2}, \ref{l5}, \ref{l7}}).
	\item \textbf{Offline Parameter Optimization and Global Solution.} A semi-definite programming (SDP) problem considering the effect of $ Q $ and $ R $ is provided for designing the sub-optimal event-triggered parameter  in the open-loop scenario (\textbf{Theorem~\ref{t2}}) and the gap is analyzed. For the closed-loop scenario, a \textit{jointly constrained biconvex} problem (\textbf{Theorem~\ref{t4}}) is derived. Furthermore, we analyze the compact set of the feasible region (\textbf{Lemma \ref{ll7}}) to ensure the boundary solutions (\textbf{Lemma~\ref{lb}}). A \textit{branch-and-bound} method (\textbf{Algorithm \ref{al2}, \ref{ag1}}) is introduced to obtain the global solution.
\end{enumerate}

	The remainder of this paper is organized as follows. The system structure and problem formulation are shown in Section \ref{s2}. The performance and the optimization problem in the open-loop case are analyzed in Section \ref{s3}. The closed-loop scenario is studied in Section \ref{s6}, where the jointly constrained biconvex programming and the branch-and-bound method are introduced. Numerical examples are given in Section \ref{s7}. Conclusions are summarized in Section \ref{sc}  and some proofs are attached in the Appendix.

	\textit{Notations: }$ \mathbb{N} $ is the set of natural numbers. $ \mathbb{R} $ and $ \mathbb{R}^{n} $ represent the set of real numbers and $ n- $dimensional column vectors, respectively.  When a matrix $ X $ is  $ n\times n  $ positive semi-definite (definite), we simply write $ X\geq0 $ ($ X>0  $) and $ X\in \mathbb{S}_{+}^{n}(\mathbb{S}_{++}^{n})
	$. For any matrix $ X $, $ {\rm tr}(X) $ and $ X^{\top} $ are its trace and the transpose, and $ \rho(X) $ is the spectrum radius of $ X $.  When  $ 	X>0$,  we use $ \ve(X)\in\mathbb{R}^{n^{2}}  $($  \veh(X)\in\mathbb{R}^{\frac{n(n + 1)}{2}}  $) to denote the  vectorization (half-vectorization) of $ X $.	The identity matrix is $ I $ and its size is determined from the content.
	$ \Pr(\cdot) $ and $ \Pr(\cdot|\cdot) $ stand for the probability and conditional probability. $ \mathbb{E}[\cdot] $ denotes the expectation of a random variable. $ f(x|y) $ denotes the probability density function (pdf) of a random variable (r.v.) $ x $ conditional on a random variable $ y $. $ \mathcal{N}(\mu,\Sigma ) $ denotes a Gaussian distribution with mean $ \mu $ and covariance matrix $ \Sigma $. The sequence $ \{\eta_{0},\eta_{1},\ldots,\eta_{k}\} $ is simplified as $ \{\eta_{k}\}_{0}^{k} $. We denote $ \{\eta_{0},\eta_{1},\ldots\} $ as $ \{\eta_{k}\}_{0}^{\infty} $. For a compact set $ \Omega $, $ \partial\Omega $ is the boundary of $ \Omega $.
	
	\section{Problem Formulation}\label{s2}
	\subsection{System structure}
	\begin{figure}[!t]
		\centering
		\includegraphics[width=0.9\columnwidth]{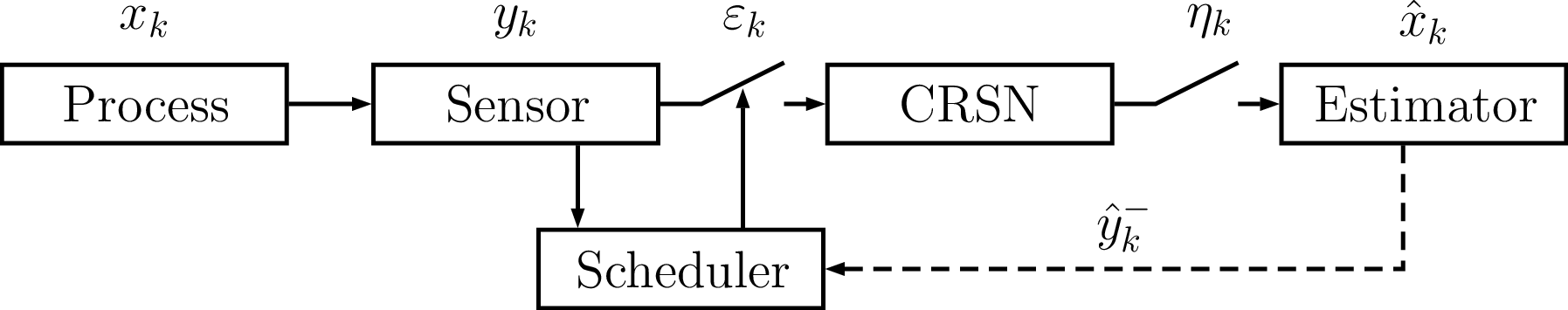}
		\caption{System structure.}
		\label{f1}
	\end{figure} 
	The system structure is shown in Fig. \ref{f1}.  		
	Consider a discrete-time linear time-invariant (LTI) process: 
	\begin{equation}\label{eq1}
	\begin{split}
	x_{k+1}&=Ax_{k}+w_{k},\\
	y_{k}&=Cx_{k}+v_{k},
	\end{split}
	\end{equation}
	where $ x_{k}\in \mathbb{R}^{n} $ is the system state, $  y_{k} \in{\mathbb{R}^{m}}  $ is the measurement vector taken by the sensor at
	time $ k $, and $ w_{k} \in \mathbb{R}^{n} $ and $ v_{k} \in \mathbb{R}^{m} $ are two independent identically distributed  (i.i.d.) zero-mean Gaussian random
	noises with covariances $ Q \geq 0$ and $ R > 0 $, respectively.
	The initial state $  x_{0} $ is a zero-mean Gaussian r.v. that
	is uncorrelated with $ w_{k} $ or $ v_{k} $ and has covariance $ \Pi_{0} \geq 0 $.   We 
	assume that $( A , \sqrt{Q} )$ is stabilizable and $(A, C)$ is detectable. 
	
	A sensor is equipped with a pre-designed event-based scheduler deciding whether to send sensor's measurement or not. Let  $ \varepsilon_{k}\in\{0,1\} $ denote the decision variable. If $ \varepsilon_k=1 $, $ y_k $ is sent; otherwise it is not sent. We consider the remote estimation problem where the sensor output is transmitted to the estimator via a CRSN. The CRSN consists of $ M $ PUs, where each PU can access one channel. When the PU is absent, a SU base station collecting all the SUs' states within its coverage could send information to the remote estimator\cite{chen2010power}. Without loss of generality, we  use the aggregated $ x_{k} $ and $ y_{k} $, thus we only need to study one SU case. Let $ \eta_{k}\in\{0,1\} $ represent the channel's availability. If $  \eta_{k}=0 $, the channel is occupied by the PU, and vice versa.
	We assume $ \eta_{k} $ evolves as an i.i.d. Bernoulli random process with $ \mathbb{E}[\eta_{k}]=\lambda\in(0,1] $, which is widely used in  \cite{saleem2014primary, ganti2002tranmission,ganti2007optimal,banaei2009throughput,gambini2008packet}.
	
	The remote estimator can identify the packet source. If the remote estimator receives no packet, it means that the channel condition is idle for the sensor to transmit packet but the scheduler decides not to send, i.e., $ \eta_{k}=1$ and $ \varepsilon_{k}=0 $. If the remote estimator receives a packet from other sensors, the channel condition is unfavorable at this time step, i.e., $ \eta_{k}=0$. In this case, there is no information of $ \varepsilon_{k} $, and we can set it as $ \varepsilon_{k}=\emptyset $. Otherwise, the remote estimator receives the packet from this sensor, i.e., $ \eta_{k}=1$ and $ \varepsilon_{k}=1$. The following information is available to the estimator at time $ k $
	\begin{align*}
	\mathcal{I}_{k}\triangleq\{\eta_{k}\}_{0}^{k}\cup\{\varepsilon_{k}\}_{0}^{k}\cup\{\eta_{k}\varepsilon_{k}y_{k}\}_{0}^{k},
	\end{align*} 
	with $ \mathcal{I}_{-1}=\emptyset $.  
	Further define the following notations which will be used in subsequent analysis:
	\begin{align*}
	&\hat{x}_{k}^{-}	\triangleq	\mathbb{E}[x_k | \mathcal{I}_{k-1}]	,\hat{y}_{k}^{-}	\triangleq	\mathbb{E}[y_k | \mathcal{I}_{k-1}],\\
	&e_{k}^{-}	\triangleq	x_{k}-	\hat{x}_{k}^{-},	
	P_{k}^{-}		\triangleq		\mathbb{E}[	e_{k}^{-}e_{k}^{-\top}],\\
	&\hat{x}_{k}\triangleq	\mathbb{E}[x_k | \mathcal{I}_{k}]	,e_{k}	\triangleq	x_{k}-	\hat{x}_{k}, 	
	P_{k}		\triangleq		\mathbb{E}[	e_{k}	e_{k}^{\top}].
	\end{align*}
	The estimates $ \hat{x}_{k}^{-} $ and $ \hat{x}_{k} $ are the \textit{a priori} and the \textit{a posteriori} MMSE state estimate, respectively. Meanwhile,  $ P_{k}^{-} $ and $ P_{k} $ are the \textit{a priori} and the \textit{a posteriori} estimation error covariance, respectively. Similarly, $ \hat{y}_{k}^{-}  $ denotes the \textit{a priori} MMSE measurement estimate.
	
	We adopt the stochastic event-triggered scheduling schemes in~\cite{han2015stochastic} as below.  At each time step,	the sensor generates an i.i.d. random variable $ \zeta_{k} $ which is uniformly distributed over $ [0,1] $, denoted as $ \zeta_{k}\sim U(0,1) $. The transmission decision by the sensor, i.e., $ \varepsilon_{k} $ follows two event-triggered criteria.
	\begin{enumerate}
		\item  Open-loop scheduler: The sensor makes the decision based on the current raw measurement $ y_{k} $, i.e.,
		\begin{equation}\label{eq2}
		\varepsilon_{k} =\left\lbrace \begin{array}{l}
		1,\text{ if } \zeta_{k}> exp(-\frac{1}{2}y_{k}^{\top}Yy_{k}), Y>0,\\ 
		0,\text{ otherwise} .
		\end{array} \right. 
		\end{equation}	
		\item Closed-loop scheduler: The sensor receives a feedback $ \hat{y}_{k}^{-} $ from the remote estimator; then the decision is based on the measurement innovation $ z_{k}\triangleq y_{k}-\hat{y}_{k}^{-} $ as
		\begin{equation}\label{eq62}
		\varepsilon_{k} =\left\lbrace \begin{array}{l}
		1,\text{ if } \zeta_{k}> exp(-\frac{1}{2}z_{k}^{\top}Zz_{k}), Z>0,\\ 
		0,\text{ otherwise} .
		\end{array} \right. 
		\end{equation}
	\end{enumerate}
	The open-loop scheduler is easier to implement since it does not require any feedback. However, open-loop schedulers cannot reduce the communication rate for unstable systems since $ \varepsilon_{k}=1 $ almost surely occurs for any given $ Y $ after a long time\cite{han2015stochastic}. Thus we need closed-loop schedulers to reduce the communication rate for unstable systems.
	\begin{remark}
		We choose these schedulers because they preserve the Gaussian property which will be exploited in Theorem \ref{t1} to obtain the linear recursion of update. This refrains from nonlinear complicated and approximate estimation using the other existing event-triggered mechanism, e.g., \cite{wu2013event} and \cite{wang2012asynchronous}. 
	\end{remark}
	
	\subsection{Problem of Interest}
	Define the average  communication rate  as 
	\begin{equation}\label{eq4}
	\gamma\triangleq\limsup\limits_{N\rightarrow\infty}\dfrac{1}{N}\sum\limits_{k=0}^{N-1}\mathbb{E}[\eta_{k}\varepsilon_{k}]=\lambda\limsup\limits_{N\rightarrow\infty}\dfrac{1}{N}\sum\limits_{k=0}^{N-1}\mathbb{E}[\varepsilon_{k}].
	\end{equation}
	Since the sequence $ \{\eta_{k}\}_{0}^{\infty} $ has no relationship with the measurement, the iteration of the error covariance is stochastic and cannot be determined offline. Therefore, we are interested in its statistical properties. Define the mean error  covariance of the system at time $ k $ as
	$ \mathbb{E}[P_{k}^{-}] $. 
	We are interested in the following problem:
	\begin{problem}\label{pi}
		\begin{equation*}
		\begin{split}
		\min\, &\gamma \\
		{\rm  s.t.} &\mathbb{E}[P_{k}^{-}]\leq M, 		
		\end{split}		
		\end{equation*}
		where $ M>0 $ is a given matrix-valued bound.
	\end{problem}
	We study two extreme cases to demonstrate that the event-triggered parameter influences $ \gamma  $ and $ \mathbb{E}[P_{k}^{-}] $. Note that if $ Y=0 $, $ \varepsilon_{k}=0 $ almost surely occurs. Therefore, if $ \gamma=0 $, the mean error covariance of the remote estimator diverges for an unstable system, i.e., $ \mathbb{E}[P_{k}^{-}]\rightarrow\infty $. On the other hand, for sufficiently large $ Y $ such that $ exp(-\frac{1}{2}y_{k}^{\top}Yy_{k})=0 $ almost surely occurs, we have $ \gamma=\lambda $. The  error covariance converges if $ \lambda>1-1/\rho(A)^{2} $~\cite{sinopoli2004kalman}. The same analysis applies to $ Z $. To avoid trivial problems, we assume this condition is satisfied in the following analysis.
	
	It is obvious that the parameter $ Y $ or $ Z $ introduces an additional degree of freedom to balance the tradeoff between the communication rate and the mean error covariance. However, it is difficult to solve Problem \ref{pi} directly since both the objective and constraint are implicit functions of $Y$ or $Z$.  One core problem lies in whether we are able to obtain the explicit expression of the communication rate and the mean error covariance in terms of $ Y $ or $ Z $. If not, we expect to find some bounds of the mean error covariance. In this paper, we will focus on the derivation of the communication rate and the mean error covariance in terms of $Y$ or $Z$ and then design the parameters to achieve a desired tradeoff. Besides, we will also explore an explicit MMSE estimator since it is also critical for the system implementation.

	\section{Open-Loop Scenario} \label{s3}
	The case using the open-loop schedulers is called the open-loop scenario. In the open-loop scenario, the main difficulty is that, due to the randomness of $ \{\eta_{k}\}_{0}^{\infty} $, the error covariance is stochastic and cannot be determined a priori. Only the mean error covariance is deduced.  It is difficult to analyze the iterative behavior of the mean error covariance  because of  the nonlinearity of the error covariance's recursion function. The influence of the event-triggered parameter on the mean error covariance is analyzed in this section.
	
	Since the open-loop scheduler \eqref{eq2} only reduces the communication rate for stable systems as Section \ref{s2}.A, we study the stable system in the open-loop case. We assume in the sequel that the system has already entered into the steady state, which implies that 
	\begin{equation*}
	P_{0}^{-}=Cov(x_{k})=\Sigma,Cov(y_{k})=\Pi,
	\end{equation*}  
	where $ \Sigma=A\Sigma A^{\top}+Q $, $ \Pi=C\Sigma C^{\top}+R $. 
	
	Given $ Y $, the average  communication rate~\cite{han2015stochastic} is 
	\begin{equation}\label{eq16}
	\begin{split}
	\gamma=\lambda\left( 1-(\det(I+\Pi Y))^{-\frac{1}{2}}\right) .
	\end{split}			
	\end{equation}
	
	Define functions $ h $, $ g_{\theta,W}:\mathbb{S}_{+}^{n}\rightarrow \mathbb{S}_{+}^{n}$ as follows:
	\begin{equation*}
	\begin{split}
	h(X)\triangleq& AXA^{\top}+Q, \\
	g_{\theta,W}(X)\triangleq &AXA^{\top}+Q-\theta AXC^{\top}(CXC^{\top}+W)^{-1}CXA^{\top}	,
	\end{split}	
	\end{equation*}
	where $ X>0 $,  $ W>0 $ and $ \theta\in(0,1] $. The function $ h $ can be interpreted as the recursive function of the estimation error covariance matrix when the channel is not available while the function $ g $ is the modified algebraic Riccati equation for the Kalman filter with intermittent observations \cite{sinopoli2004kalman}. If $ \theta=1 $, $ g_{1,W}$ will be written as $ g_{W}$ for brevity. 
	The propositions of function $ g_{\theta,W}(X) $ are shown in Appendix \ref{a2}.
	
	\begin{theorem}\label{t1}
		The MMSE estimate under an open-loop scheduler is computed as follows. 	
		Start from the initial condition $ \hat{x}_{0}^{-}=0 $ and $ P_{0}^{-}=\Sigma $.\\		
		Measurement Update:
		\begin{equation} \label{eq7}
		\begin{split} 	
		K_k		&=	\eta_{k}P_{k}^{-}C^{\top}\big(CP_{k}^{-}C^{\top} + R+(1-\varepsilon_{k})Y^{-1}\big)^{-1} 	\\
		\hat{x}_{k}	& =\hat{x}_{k}^{-}+\varepsilon_{k}K_ky_{k}-K_{k}\hat{y}_{k}^{-}=(I-K_{k}C)\hat{x}_{k}^{-}+\varepsilon_{k}K_ky_{k}, 	\\
		P_{k}			&=	P_{k}^{-} -K_k	C P_{k}^{-},
		\end{split}
		\end{equation} 		
		Time Update:
		\begin{equation} \label{q6}
		\begin{split}
		\hat{x}_{k+1}^{-}	=	A\hat{x}_{k},
		P_{k+1}^{-}		=	h(P_{k}).
		\end{split}
		\end{equation}
	\end{theorem}
	\begin{proof}
		See Appendix \ref{a1}.	
	\end{proof}	
	
	The proof of this theorem uses the  Gaussian property of the distribution proved in~\cite{han2015stochastic}. By exploiting the Gaussian property, the recursion of the update is linear, which reduces computational complexities. 
	
	By exploiting the concavity, monotonicity and limit property, the asymptotic upper and lower bounds on $ \mathbb{E}[P_{k}^{-}] $ are shown as Lemma~\ref{l2}.

	\begin{lemma}\label{l2}		
		The mean error covariance $ \mathbb{E}[P_{k}^{-}]  $ satisfies
		\begin{equation*}
		g_{R_{1}}^{k}(\Sigma)\leq\mathbb{E}[P_{k}^{-}]\leq g_{\lambda,R+Y^{-1}}^{k}(\Sigma),
		\end{equation*}
		where $ R_{1}^{-1}=\gamma R^{-1}+(\lambda-\gamma)(R+Y^{-1})^{-1} $.
		
		The asymptotic upper and lower bounds on $ \mathbb{E}[P_{k}^{-}]  $ are
		\begin{equation}\label{e19}
		\underline{X}_{ol}\leq\liminf\limits_{k\rightarrow\infty}\mathbb{E}[P_{k}^{-}]\leq\limsup\limits_{k\rightarrow\infty}\mathbb{E}[P_{k}^{-}] \leq \bar{X}_{ol}, 
		\end{equation}
		where  $ \underline{X}_{ol}>0$ is the unique solution to	
		$ \underline{X}_{ol}=g_{R_{1}}(\underline{X}_{ol}) $,
		and $ \bar{X}_{ol}>0 $ is the unique solution to	
		$ \bar{X}_{ol} =g_{\lambda,R+Y^{-1}}(\bar{X}_{ol}) $.
		
		For all schedules satisfied \eqref{eq2}, we  obtain that$  \underline{X}_{ol}\geq X_{0}$, where $ X_{0}>0 $ is the unique solution to
		\begin{equation}\label{e15}
		X_{0}=g_{R/\lambda}(X_{0}) .
		\end{equation}		
	\end{lemma}
	\begin{proof}
		See Appendix \ref{a3}.
	\end{proof}

	\begin{figure*}[ht]
		\hrulefill	
		\normalsize
		
		\begin{equation}\label{eq10}
		\Psi(S,Y)\triangleq\left[ \begin{array}{ccccc}
		S& \sqrt{\lambda}SA & \sqrt{1-\lambda}SA & S & 0 \\ 
		\sqrt{\lambda}A^{\top}S& S+C^{\top}R^{-1}C & 0 & 0 &  C^{\top}R^{-1}\\ 
		\sqrt{1-\lambda}A^{\top}S& 0 & S & 0 & 0 \\ 
		S& 0 & 0 & Q^{-1} & 0 \\ 
		0& R^{-1}C & 0 & 0 & Y+R^{-1}
		\end{array} \right]	
		\end{equation}
		\hrulefill
		\vspace*{4pt}
	\end{figure*}
	\begin{remark}
		By applying the information filtering and exploiting the convexity of $ X^{-1} $, we obtain a different lower bound on $ \mathbb{E}[P_{k}^{-}]  $, i.e., $ \underline{X}_{ol} $. We plot it with respect to (w.r.t.) $ \gamma $ in Fig.~\ref{f4}. It is different from the lower bound derived in~\cite{sinopoli2004kalman} which is denoted as $ X_{p} $, where $ X_{p}=(1-\lambda)AX_{p}A^{\top}+Q  $. The matrix~$ X_{0}  $ is the lower bound of $ \mathbb{E}[P_{k}^{-}] $ for all  schedulers. When $ \lambda=1 $, the lower bound derived in our paper is larger than $ X_{p} $, i.e.,  $ \underline{X}_{ol}>Q=X_{p} $. For scalar systems, we can choose $\max\{\underline{X}_{ol},X_{p}\}  $ to be the lower bound.
	\end{remark}
	
	From the above analysis, we relax Problem \ref{pi} to bound the asymptotic upper bound on the  mean error covariance, i.e.,~$ \bar{X}_{ol}  $. 
	\begin{problem}\label{q2}
		\begin{equation*}
		\begin{split}
		\min\limits_{Y}\, &\gamma \\
		{\rm s.t.}\, & {Y\geq 0},\bar{X}_{ol} \leq M. 
		\end{split}		
		\end{equation*}	
	\end{problem}	
	We observe that for the scaler case ($Y\in\mathbb{R}$), the above problem can be easily solved by convex programming. However, for the general vector cases, Problem $2$ is not convex since $(5)$ is a log concave function of matrix $Y$. We need the following lemma to study the general vector case.
	\begin{lemma}\cite[Lemma 2]{han2015stochastic}\label{l3}	
		Given $ \gamma $ in \eqref{eq16}, $ \Pi,Y\in\mathbb{S}_{+}^{n} $, the following inequality holds:
		\begin{equation}\label{eq28}
		f_{1}\big({\rm tr }(\Pi Y)\big)\leq\gamma\leq f_{2}\big({\rm tr} (\Pi Y)\big),
		\end{equation}
		where $ 	f_{1}(x)=\lambda(1-(1+x)^{-\frac{1}{2}}),f_{2}(x)=\lambda(1-\exp(-\frac{1}{2}x))$.
		The equality is only satisfied when $ {\rm tr }(\Pi Y)=0 $.		 
	\end{lemma}
	
	Using Lemma \ref{l3}, the objective of Problem \ref{pi} is bounded by two increasing functions. Thus, it can be relaxed into $ \min {\rm tr}(\Pi Y)$. 
	\begin{problem}\label{q3}
		\begin{equation*}
		\begin{split}
		\min\limits_{Y}\,&{\rm tr} (\Pi Y)\\
		{\rm s.t.} \, &{Y\geq 0}, \bar{X}_{ol} \leq M,  \bar{X}_{ol} =g_{\lambda,R+Y^{-1}}(\bar{X}_{ol}).
		\end{split}		
		\end{equation*}
	\end{problem}
	
	We  transform Problem \ref{q3} into an SDP problem using Theorem \ref{t2}.	
	\begin{theorem}\label{t2}
		Problem \ref{q3} is equivalent to
		\begin{equation}\label{eq31}
		\begin{split}
		&\min \limits_{S,Y}{\rm tr} (\Pi Y)\\
		{\rm s.t. }\,& \Psi(S,Y)\geq 0,	\left[ \begin{array}{cc}
		S& I \\ 
		I& M
		\end{array} \right] \geq 0,Y\geq0,
		\end{split}	
		\end{equation}
		where $  \Psi(S,Y)	 $ is defined as \eqref{eq10}.
	\end{theorem}
	\begin{proof}
		The proof mainly follows two steps. First we prove equivalent LMIs to replace the implicit function $ g_{\lambda, R+Y^{-1}} $. Then the function $ g_{\lambda, R+Y^{-1}}(X)\leq X $ is transformed into an equivalent SDP constraint. The details are shown in Appendix \ref{a4}.
	\end{proof}	
	\begin{remark}
		 Sinopoli et al.~\cite{sinopoli2004kalman} also derived an SDP constraint from $ g_{\lambda, W}(X)\leq X $, but they neglected the influence of $ Q $ and $ W $ through relaxation. However, in our case, since $ W $ corresponds to the decision variable $ Y $, we cannot eliminate the influence.  
	\end{remark}

	Define  $ \gamma^{*} $ as the communication rate with optimal $ Y^{*} $ in Problem \ref{q3}.	 Let the optimal solution to Problem \ref{q2} be $ Y^{opt} $ and the minimum  objective be $ \gamma^{opt} $. Define the gap $ \kappa $ as $ 	\kappa \triangleq\gamma^{*}-\gamma^{opt} $. By \eqref{eq28}, one has
	\begin{equation}\label{eq30}
	0<\kappa<\lambda\left( \left( 1+{\rm tr}(\Pi Y^{*})\right) ^{-\frac{1}{2}}-\left( \det(I+\Pi Y)\right) ^{-\frac{1}{2}}\right) .
	\end{equation} 
	Fig. \ref{fp1} shows the relationship between the problems.

	\begin{figure}[!t]
		\centering
		\includegraphics[width=0.9\columnwidth]{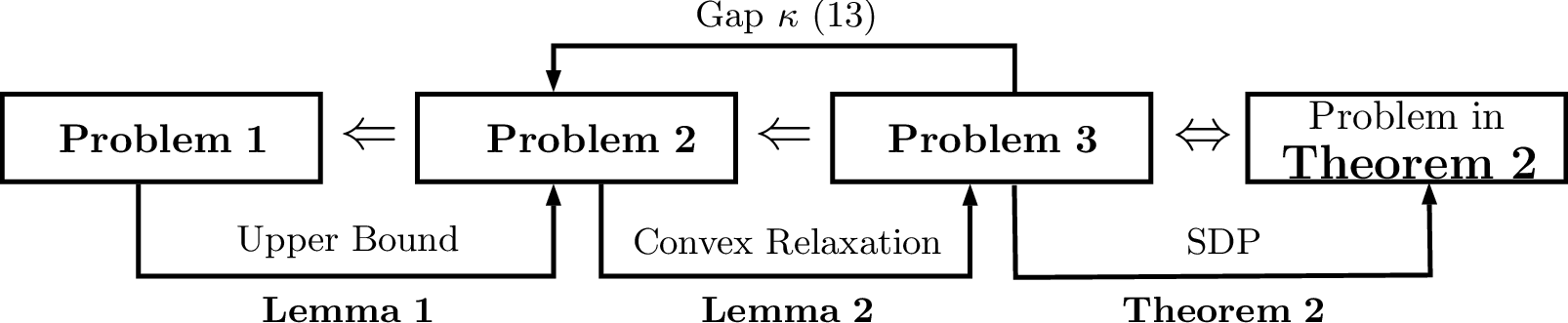}
		\caption{The relationship between each problem in the open-loop case.			
			(``$ \Leftarrow $'' denotes that the optimal solution to this problem will satisfy the second problem's constraints and ``$ \Leftrightarrow $'' represents these two problems are equivalent.)}
		\label{fp1}
	\end{figure}

	\section{closed-loop Scenario} \label{s6}
	The case which feeds $ \hat{y}_{k}^{-} $ back from the remote side and then adopts closed-loop schedulers to trigger system is called closed-loop scenario. In closed-loop scenario, the same difficulty  lying in analyzing the mean error covariance  remains since the error covariance cannot be determined a priori. In addition, different from Section \ref{s3}, the communication rate herein depends on the event-triggered parameter as well as the realization of the error covariance; therefore it is not able to be determined a priori too.
	
	\subsection{Performance Analysis and Problem Reformulation}		
	\begin{theorem}\label{t3}
		The MMSE estimate under a closed-loop scheduler is as follows. 	
		Start from the initial condition $ \hat{x}_{0}^{-}=0 $ and $ P_{0}^{-}=\Pi_{0} $.\\
		Measurement Update:
		\begin{equation} 
		\begin{split} 	
		K_k		&=	\eta_{k}P_{k}^{-}C^{\top}\big(CP_{k}^{-}C^{\top} + R+(1-\varepsilon_{k})Z^{-1}\big)^{-1} 	\\
		\hat{x}_{k}	& = \hat{x}_{k}^{-}+\varepsilon_{k}K_kz_{k}, 	\\
		P_{k}			&=	P_{k}^{-} -K_k	C P_{k}^{-},
		\end{split}
		\end{equation}		
		Time Update:
		\begin{equation} \label{e33}
		\begin{split}
		\hat{x}_{k+1}^{-}=	A\hat{x}_{k},	P_{k+1}^{-}		=	h(P_{k}).
		\end{split}
		\end{equation}
		
	\end{theorem}
	\begin{proof}
		See Appendix \ref{at3}.
	\end{proof}

	Similar to equation \eqref{eq16}, we obtain the average communication rate in a closed-loop scenario as the following lemma.
	\begin{lemma} \label{l4}
		The average  communication rate  $ \gamma $ under a closed-loop scheduler \eqref{eq62} is 
		\begin{equation}\label{eq70}
		\begin{split}
		\gamma=&\lambda\mathbb{E}\left[ 1-\left( \text{det}\left( I+(CP_{k}^{-}C^{\top}+R)Z\right) \right)^{-\frac{1}{2}}\right]  \\
		\leq&\lambda\left( 1-\left( \text{det}\left( I+(C\mathbb{E}[P_{k}^{-}]C^{\top}+R)Z\right) \right)^{-\frac{1}{2}}\right) .
		\end{split}		
		\end{equation}
	\end{lemma}
	
	\begin{proof}
		See Appendix \ref{a5}.			
	\end{proof} 
	
	From the above lemma, once the upper bound of $ \mathbb{E}[P_{k}^{-}] $ is obtained in the closed-loop scenario, we derive the upper bound of the average communication rate. The asymptotic upper and lower bounds on $ \mathbb{E}[P_{k}^{-}] $ under the closed-loop scheduler are shown in Lemma \ref{l5} and Lemma \ref{l7}.
	
	\begin{lemma}\label{l5}		
		The mean error covariance $ \mathbb{E}[P_{k}^{-}]  $ is bounded by
		\begin{equation*}
		\mathbb{E}[P_{k}^{-}]\leq g_{\lambda,R+Z^{-1}}^{k}(\Sigma).
		\end{equation*}		
		
		The asymptotic upper bound on $ \mathbb{E}[P_{k}^{-}]  $ is
		\begin{equation*}
		\limsup\limits_{k\rightarrow\infty}\mathbb{E}[P_{k}^{-}] \leq \bar{X}_{cl}, 
		\end{equation*}
		where  $ \bar{X}_{cl}>0 $  is the unique solution to		
		$ \bar{X}_{cl}=g_{\lambda,R+Z^{-1}}(\bar{X}_{cl}) $. 	
	\end{lemma}

	From Lemma \ref{l4} and Lemma \ref{l5}, denote the upper bound of the average  communication rate as
	\begin{equation}\label{eq17}
	\bar{\gamma}\triangleq\lambda\left( 1-\left( \text{det}\left( I+(C\bar{X}_{cl}C^{\top}+R)Z\right) \right)^{-\frac{1}{2}}\right).
	\end{equation}
	\begin{remark}
		From \eqref{eq16} and \eqref{eq17}, given the same quality constraints, one has the upper bound of the average communication rate using closed-loop schedules is smaller than the average communication rate by open-loop schedules. Thus, the closed-loop schedulers outperform the open-loop schedules.
	\end{remark}
	
	Substituting the upper bound of the communication rate, we further obtain the asymptotic lower bound on $ \mathbb{E}[P_{k}^{-}] $ as follows. 
	\begin{lemma}\label{l7}
		The mean error covariance $ \mathbb{E}[P_{k}^{-}]  $ is bounded by
		\begin{equation*}
		g_{R_{1}}^{k}(\Sigma)\leq\mathbb{E}[P_{k}^{-}]
		\end{equation*}
		where $ R_{1}^{-1}=\bar{\gamma} R^{-1}+(\lambda-\bar{\gamma})(R+Z^{-1})^{-1} $.
		
		The asymptotic lower bound on $ \mathbb{E}[P_{k}^{-}]  $ is
		\begin{equation*}
		\underline{X}_{cl}\leq\liminf\limits_{k\rightarrow\infty}\mathbb{E}[P_{k}^{-}], 
		\end{equation*}
		where 	$ \underline{X}_{cl}>0$ is the unique solution to		
		$ 	\underline{X}_{cl}=g_{R_{1}}(\underline{X}_{cl}) $.

		Meanwhile, one has $ \underline{X}_{cl}\geq X_{0} $, where $ X_{0}  $ from \eqref{e15} 	is the lower bound of $ \mathbb{E}[P_{k}^{-}] $ for all  schedulers satisfied \eqref{eq62}.
	\end{lemma}
	\begin{proof}
		The proof of Lemma \ref{l5} and Lemma \ref{l7} is similar to that of Lemma \ref{l3} as shown in Appendix \ref{a3}. Hence, we omit this part.
	\end{proof}

	In the closed-loop scenario, since there is no closed-form of $ \gamma $ given $ Z $, we relax the objective in Problem \ref{pi}  by the upper bound $ \bar{\gamma} $ and the asymptotic upper bound on the mean error covariance $ \bar{X}_{cl} $. 	
	\begin{problem}\label{q4}
		\begin{equation*}
		\begin{split}
		\min \limits_{Z}\, & 	\bar{\gamma}  \\
		{\rm s.t.} \, & \bar{X}_{cl} \leq M,	Z\geq	  0.
		\end{split}		
		\end{equation*}	 
	\end{problem}
	We also take Lemma \ref{l3} to relax Problem \ref{q4} to  Problem \ref{p4} which is equivalent to Problem \eqref{e44} in Theorem \ref{t4} as further proved. Fig. \ref{fp2} shows the relationship between each problem directly.
	\begin{figure}[!t]
		\centering
		\includegraphics[width=0.9\columnwidth]{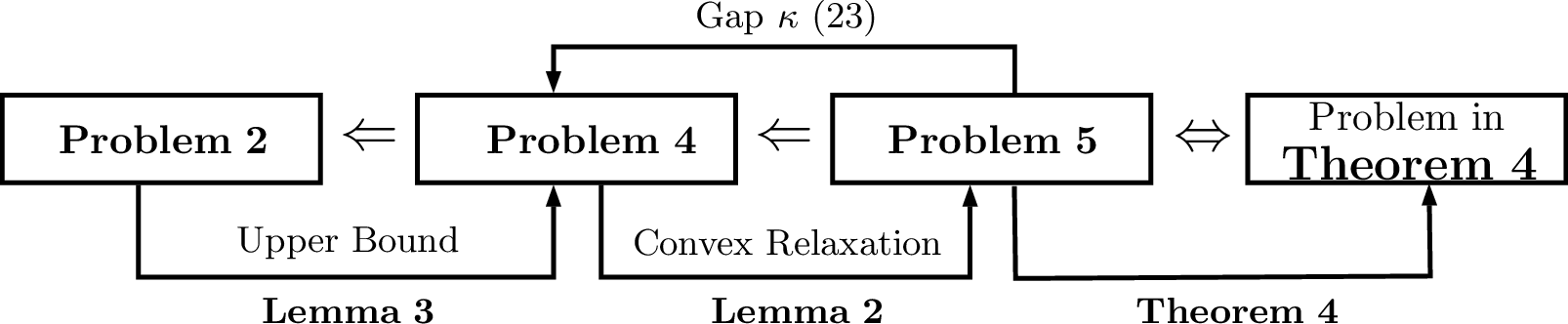}
		\caption{The relationship between each problem the closed-loop case.}
		\label{fp2}
	\end{figure} 
	\begin{problem}\label{p4}
		\begin{equation*}
		\begin{split}
		\min \limits_{T,Z}\, & 	{\rm tr} (TZ)  \\ 
		{\rm s.t.} \, & \bar{X}_{cl} \leq M,	C\bar{X}_{cl} C^{}+R\leq T, \bar{X}_{cl}=g_{\lambda,R+Z^{-1}}(\bar{X}_{cl}),	Z\geq	  0.
		\end{split}		
		\end{equation*}	 
	\end{problem}
	
	\begin{theorem}\label{t4}
		Problem \ref{p4} is equivalent to
		\begin{equation}\label{e44}
		\begin{split}
		\min\limits_{Z,X,S}\, & f(X,Z)\triangleq {\rm tr} \left((CX C^{\top}+R)Z\right)\\
		{\rm s.t.} &  (X,Z,S)\in \mathcal{S},
		\end{split}		
		\end{equation}	
		where $ \mathcal{S} $ represent the constraints of the decision variables, i.e.,
		\begin{equation}\label{q19}	
		\begin{split}
		\mathcal{S}\triangleq &\{ (X,Z,S):\\
		&\left[ \begin{array}{cc}
		X& I \\ 
		I& S
		\end{array} \right] \geq 0,\left[ \begin{array}{cc}
		S& I \\ 
		I& M
		\end{array} \right] \geq 0,	\Psi(S,Z) \geq 0,	Z\geq 0 \}.
		\end{split}
		\end{equation}	
	\end{theorem}
	\begin{proof}
		See Appendix \ref{a6}.		
	\end{proof}

	\subsection{Jointly Constrained Biconvex Programming}
The objective of Problem \eqref{e44} is a bilinear function w.r.t. the psd cone $ X $ and $ Z $. Linear operations, such as matrix vectorization, are applied to convert it into a general form of the jointly biconvex program \cite{al1983jointly}. 
Define
\begin{equation*}\label{eq36}
\tilde{x}\triangleq \ve(X),	\tilde{y}\triangleq \ve(Z)\footnote{For symmetric matrix, half vectorization is a bijection of vectorization \cite{henderson1979vec}. To simplify the notation, we use vectorization throughout the paper. In real implementation of the algorithms, half vectoeriation representation is used to reduce the dimension of the optimization variables.}.
\end{equation*}
It is obvious that $ \tilde{x} $ and $\tilde{y}  $ are bijections of $ X $ and $ Z $.  Thus, the optimization parameters can be changed to $ (\tilde{x} ,\tilde{y}, S) $ and the objective function is a bilinear function w.r.t. $ \tilde{x}$ and $\tilde{y} $ as follows
\begin{equation}\label{e21}
f(S,P)=\phi(\tilde{x},\tilde{y})=\tilde{x}^{\top}G\tilde{y}+g(\tilde{y}),
\end{equation}
where
\begin{equation}
G=C^{\top} \otimes C^{\top},g(\tilde{y})={\rm tr}(RZ).
\end{equation}

\begin{lemma}[Boundary Solution]\cite{al1983jointly}[Theorem 1] \label{lb}The jointly biconvex Problem \eqref{e44} has boundary solutions if the feasible region is compact.
\end{lemma}

The jointly constrained biconvex problem can be solved numerically by some methods. Branch-and-bound (B\&B) algorithm proposed by \cite{al1983jointly} is one of them and produces a global optimal solution. The B\&B algorithm splits the feasible region into several subregions and produces an increasing convex underestimator with an associated decreasing upper bound in the subregions. With the finer splitting, it is proved that the limit of the convex underestimator converges to the limit of the upper bound; thus the global solution is obtained.

Note that the feasible region being compact is a necessary condition to implement the jointly constrained biconvex programming. The main difficulty remains that the initial set $ \mathcal{S} $ is unbounded, i.e., there is no obvious upper bound of the elements of matrix $ Z $ ($  \tilde{y} $) in the constraints. The upper bound of $ Z $ is from the objective which aims to minimize $ {\rm tr} \left((CX C^{\top}+R)Z\right)  $. Therefore, we derive a necessary condition w.r.t. $ Z $ for the optimal solution as an upper bound requirement. The details are shown in Lemma \ref{ll7}.
	\begin{lemma}\label{ll7}
		The optimal solution $ (X^{*},Z^{*}) $ belongs to the set $ 	\Omega\triangleq\{(X,Z): j_{ij}\leq X_{ij}\leq J_{ij},	d_{ij} \leq  Z_{ij}\leq z^{*}\} $,	where
		\begin{align*}
		j_{ij}= &\left\lbrace \begin{array}{l}
		(X_{0})_{ii},\text{ if } i=j, \\ 
		-(M_{ii}M_{jj})^{\frac{1}{2}}, \text{else},
		\end{array} \right. 	J_{ij}=
		(M_{ii}M_{jj})^{\frac{1}{2}},\\
		d_{ij}= &\left\lbrace \begin{array}{l}
		0,\text{ if } i=j, \\ 
		-z^{*}, \text{else},
		\end{array} \right.	 
		\end{align*}
		and $ z^{*} $ satisfies the following optimization problem
		\begin{equation}\label{e38}
		\begin{split}
		z^{*}=&\min\limits_{Z,S}\dfrac{{\rm tr}\left((CM C^{\top}+R)Z\right)}{{\rm tr} (CX_{0} C^{\top}+R)}\\
		{\rm s.t.} & \left[ \begin{array}{cc}
		S& I \\ 
		I& M
		\end{array} \right] \geq 0,	\Psi(S,Z) \geq 0,	Z\geq 0.
		\end{split}		
		\end{equation}
	\end{lemma}
	\begin{proof}
		See Appendix \ref{a7}.		
	\end{proof}
	Adding a linear constraint $ \tilde{z}=G\tilde{y} $ to those defining $ \mathcal{S} $, the bounds
	\begin{equation}\label{e27}
	\begin{split}
	n_{k}=&\min\{(G\tilde{y})_{k}:d_{ij} \leq  Z_{ij}\leq z^{*}\},\\
	N_{k}=&\max\{(G\tilde{y})_{k}:d_{ij} \leq  Z_{ij}\leq z^{*}\} 
	\end{split}	
	\end{equation}
	replace the bounds on $ \tilde{y}  $ in defining $ \Omega $.
	
Below we detail how to construct an increasing convex underestimator with an associated decreasing upper bound for the \textit{jointly constrained biconvex} problem. We first introduce a convex envelope of $ \tilde{x}^{\top}\tilde{z} $ over $ \Omega $. It is the pointwise supremum of all convex functions which underestimate $ \tilde{x}^{\top}\tilde{z} $ over $ \Omega $, denoted by $ \Vex_{\Omega}(\tilde{x}^{\top}\tilde{z} ) $. The main results are from \cite{al1983jointly}.
	\begin{lemma}[Convex Envelop]\cite[Corollary]{al1983jointly} \label{l8}
		If $ x,y\in\mathbb{R}^{n} $ and $ (x,y)\in\Omega $, where
		$	\Omega=\{(x,y):t\leq x\leq T,
		d \leq y\leq D\},
		\Omega_{i}=\{(x_{i},y_{i}):t_{i}\leq x_{i}\leq T_{i},
		d_{i} \leq y_{i}\leq D_{i}\},$
		\begin{align*}
		\Vex_{\Omega}(x^{\top}y )=&\sum\limits_{i=1}^{n} \Vex_{\Omega_{i}}(x_{i}y_{i} ),\\
		\Vex_{\Omega_{i}}(x_{i}y_{i} )\
		=&\max\{d_{i}x_{i}+t_{i}y_{i}-d_{i}t_{i},
		D_{i}x_{i}+T_{i}y_{i}-T_{i}D_{i}\}.
		\end{align*}
		Moreover, $ x^{\top}y\geq\Vex_{\Omega}(x^{\top}y ) $ for all $ (x,y)\in\Omega $, and the $ ``= ''$ holds iff $ (x,y)\in\partial\Omega $.
	\end{lemma}
	
	Let $ \psi^{1}(\tilde{x},\tilde{y})= \Vex_{\Omega^{1}}(\tilde{x}^{T}\tilde{z} )+g(\tilde{y})$ and $\Omega^{1}=\Omega  $ from Lemma \ref{ll7}. Note that $ \psi^{1} $ is the convex underestimator of $ \phi $ and, furthermore, agrees with $ \phi $ on $ \partial\Omega $. 
	Solving the convex problem	 $ \upsilon^{1}=\min\psi^{1}(\tilde{x},\tilde{y}),{\rm s.t.}  (X,Z,S)\in \mathcal{S}\cap\Omega^{1} $ (denoted as $ \mathcal{P}^{1} $) yields the optimal value $ \upsilon^{1}(\upsilon^{11})=\psi^{1}(\tilde{x}^{1},\tilde{y}^{1}) $. If $ \upsilon^{1}=\Upsilon^{1}$, where $ \Upsilon^{1}=\phi(\tilde{x}^{1},\tilde{y}^{1}) $, $ (\tilde{x}^{1},\tilde{y}^{1})  $ is a solution to \eqref{e44}.
	Otherwise, one has
	\begin{equation}\label{eq43}
	\Delta_{i}^{1}=\tilde{x}^{1}_{i}\tilde{z}^{1}_{i}-\Vex_{\Omega_{i}}(\tilde{x}^{1}_{i}\tilde{z}^{1}_{i})>0, \text{ for some  } i.
	\end{equation}
	We choose the index $\rm I $ which produces the largest difference $ \Delta_{i}^{1} $, and split the $ \rm I- $th rectangle into four subrectangles $ \Omega^{2t} $
	according to the rule illustrated in Fig \ref{fr}.
	\begin{figure}[!t]
		\centering
		\includegraphics[width=2in]{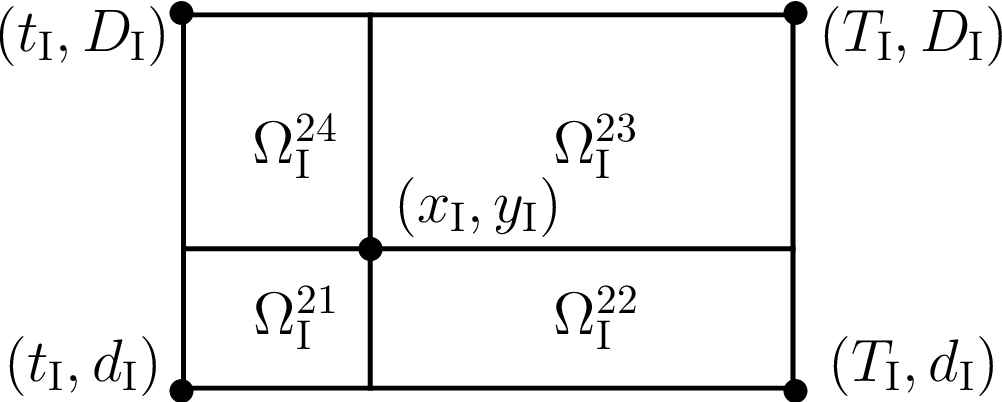}	
		\caption{Splitting $ \Omega_{\rm I} $ at stage $ 1 $.}
		\label{fr}
	\end{figure}
	
	The result of this splitting serves to set up four new subproblems, $ \min\psi^{2t}(\tilde{x},\tilde{y}),{\rm s.t.}  (X,Z,S)\in \mathcal{S}\cap\Omega^{2t} $ (denoted as $ \mathcal{P}^{2t} $) at stage $ 2 $, where $ \psi^{2t}(\tilde{x},\tilde{y})= \Vex_{\Omega^{2t}}(\tilde{x}^{T}\tilde{z} )+g(\tilde{y})$. Note that each of the subproblem $ \mathcal{P}^{2t} $ is feasible since $ (\tilde{x}^{1},\tilde{y}^{1})\in \mathcal{S}\cap\Omega^{2t}$. Moreover, by the construction of $ \psi^{2t} $, one has the minimum solution of the convex underestimator at stage $ 2 $ is larger than that at stage $ 1 $. More general, at stage $ k $, the convex problem denoted as Problem $ \mathcal{P}^{kt} $ follows
	\begin{equation}\label{e24}
	(\mathcal{P}^{kt} )\min\limits_{X,Z,S}	\psi^{kt}(\tilde{x},\tilde{y}), {\rm s.t.}  (X,Z,S)\in \mathcal{S}\cap\Omega^{kt},
	\end{equation}
	where $ \psi^{kt}(\tilde{x},\tilde{y})= \Vex_{\Omega^{kt}}(\tilde{x}^{T}\tilde{z} )+g(\tilde{y}) $. The optimal solution to $(\mathcal{P}^{kt} )  $ is $\upsilon^{kt}= \psi^{kt}(\tilde{x}^{kt},\tilde{y}^{kt}) $ and $ \Upsilon^{kt}=\phi(\tilde{x}^{kt},\tilde{y}^{kt}) $. 
	The decreasing upper bound at stage $ k $ is
	$\Upsilon^{k}\triangleq\min\{\phi: \text{all the visited optimal solution at stage  }k\}$,
	and may be expressed recursively as
	\begin{equation}\label{q44}
	\Upsilon^{k}=\min\{\Upsilon^{k-1},\Upsilon^{k1},\Upsilon^{k2},\Upsilon^{k3},\Upsilon^{k4}\}\leq\Upsilon^{k-1} .	
	\end{equation} 
	If $ \upsilon^{kt}> \Upsilon^{k}$, the subspace will be eliminated from the further consideration. Therefore, we only record the boundary $ \Omega^{kt} $ and the optimal solution $ (\tilde{x}^{kt},\tilde{y}^{kt}) $ as an open node $ (k,t) $ if $ \upsilon^{kt}\leq \Upsilon^{k}$. The increasing lower bound at stage $ k $ is
	\begin{equation}\label{q27}
	\nu^{k}\triangleq\min\{\nu^{lj}: \text{ node } (l,j) \text{ is open at stage }k\}.
	\end{equation}
	
	Moving from stage $ k $ to stage $ k+1 $ involves the selection of an open node (lines $ 1-3 $), and the creation of four new nodes from that node (lines $ 4-9 $) as shown in Algorithm \ref{al2}.
	\begin{algorithm}[t]
		\caption{Moving from stage $ k $ to stage $ k+1 $}	
		\begin{algorithmic}[1]
			\STATE Select an open node whose lower bound $ \nu^{lj}=\nu^{k} $;	
			\STATE $ \Omega\leftarrow \Omega^{lj} $, $ (\tilde{x}^{k},\tilde{z}^{k}) \leftarrow(\tilde{x}^{lj},G\tilde{y}^{lj})$;
			\STATE Erase $ (l,j)  $ from the open node storage;
			\FOR{$ i=1:n^2 $} 		
			\STATE $ 	\Delta_{i}^{k}\leftarrow\left[ \tilde{x}^{k}_{i}\tilde{z}^{k}_{i}-\Vex_{\Omega_{i}}\left( \tilde{x}^{k}_{i}\tilde{z}^{k}_{i} \right) \right]  $;
			\ENDFOR
			\algorithmiccomment{\% We only need to compare $ n(n+1)/2 $ elements if we use half-vectorization}
			\STATE $ {\rm I}\leftarrow\arg\max\limits_{i}\Delta_{i}^{k} $, $ k\leftarrow k+1 $;
			\STATE Initialization: $ \Omega^{k1},\Omega^{k2}, \Omega^{k3},\Omega^{k4}\leftarrow\Omega$;
			\STATE Update $  \Omega_{\rm I}^{k1},\Omega_{\rm I}^{k2}, \Omega_{\rm I}^{k3},\Omega_{\rm I}^{k4} $ by
			\begin{equation*}
			\begin{split}			
			(j^{k1}_{\I}, J^{k1}_{\I},n_{\I}^{k1},N_{\I}^{k1})\leftarrow&(j_{\I},\tilde{x}_{\I}^{k},n_{\I},\tilde{z}_{\I}^{k}),\\
			(j^{k2}_{\I}, J^{k2}_{\I},n_{\I}^{k2},N_{\I}^{k2})\leftarrow&(\tilde{x}_{\I}^{k},J_{I},n_{\I},\tilde{z}_{\I}^{k}),\\
			(j^{k3}_{\I}, J^{k3}_{\I},n_{\I}^{k3},N_{\I}^{k3})\leftarrow&(\tilde{x}_{\I}^{k},J_{\I},\tilde{z}_{\I}^{k},N_{\I}),\\
			(j^{k4}_{\I}, J^{k4}_{\I},n_{\I}^{k4},N_{\I}^{k4})\leftarrow&(j_{\I},\tilde{x}_{\I}^{k},\tilde{z}_{\I}^{k},N_{\I}).	
			\end{split}
			\end{equation*}	
		\end{algorithmic}
		\label{al2}
	\end{algorithm}
	
	Since the procedure converges to a globally optimal solution~\cite{al1983jointly}, once we have any $ \upsilon^{k}=\Upsilon^{k} $, the optimal solution is obtained as $ (\tilde{x}^{k},\tilde{y}^{k})  $. Due to the computational consideration, the algorithm can be terminated at a prespecified $ \epsilon $ degree of accuracy whenever $ 	\upsilon^{k}\geq\Upsilon^{k}-\epsilon $. 
	The algorithm is summarized in Algorithm \ref{ag1}. 
	
	\begin{algorithm}[t]
		\caption{$ \epsilon-$accuracy B\&B Algorithm}
		\textbf{Input:} $ A,C,Q,R,W,U,\Omega $; \\
		\textbf{Output:} $ X^{*},Z^{*},\upsilon^{*}$
		\begin{algorithmic}[1]		
			\STATE  Initialize: $ \Omega^{1}\leftarrow\Omega $, $ k\leftarrow1 $;;
			\STATE Solve the convex problem  $ (\mathcal{P}^{1}) $  and obtain the optimal solution $ (\tilde{x}^{1},\tilde{y}^{1}) $ with the lower bound $ \upsilon_{1} $ and the upper bound $ \Upsilon^{1} $;
			\WHILE {$ 	\upsilon^{k}<\Upsilon^{k}-\epsilon$}		
			\STATE Move from stage $ k $ to stage $ k+1 $ by Algorithm \ref{al2};		
			\FOR{$ t=1:4  $}
			\STATE Solve each of the four problem $ (\mathcal{P}^{kt})  $ in turn and obtain the point $ (\tilde{x}^{kt},\tilde{y}^{kt}) $ with the value $ \upsilon^{kt} $ and $ \Upsilon^{kt} $;		
			\item For $ \upsilon^{kt}\leq \Upsilon^{k} $, record it as an open node $ (k,t) $;
			\ENDFOR
			\STATE Update $ \Upsilon^{k} $ by \eqref{q44} and $\upsilon^{k} $ by \eqref{q27};		
			\ENDWHILE			
			\label{code:recentEnd}
			\STATE Let $ v^{*}=\upsilon^{k} $, $\ve(X^{*})=\tilde{x}^{kt} $ and $\ve(Z^{*})=\tilde{y}^{kt}$.
		\end{algorithmic}
		\label{ag1}
	\end{algorithm}
	By the same method in the open-loop case, the optimality gap 
	\begin{equation*}
	\begin{split}
	\kappa=& \bar{\gamma}^{*}-\bar{\gamma}^{opt}\\
	=&\lambda\left( \big(1+\upsilon^{*} )^{-\frac{1}{2}}-\left( \det\left( I+(CX^{*} C^{\top}+R)Z^{*}\right) \right) ^{-\frac{1}{2}}\right).
	\end{split}	
	\end{equation*}
	
	\section{simulation} \label{s7}	
	\subsection{Policy Assessment}\label{c1}
	We consider a scalar stable system with parameters $ A =0.8,C=1,Q=1$, $R=1$ and   $ \lambda= 0.8 $. We compare our stochastic event-triggered schedulers with two other offline schedulers as follows.
	\begin{enumerate}
		\item Random offline scheduler: The sensor transmits packets with probability    $ \frac{\gamma}{\lambda }$ at each time step in random scheduling;
		\item Periodic offline scheduler: The sensor sends the data using the optimal offline periodic scheduling~\cite{shi2011sensor} with rate $  \frac{\gamma}{\lambda } $.
	\end{enumerate}  
	We adopt the Monte Carlo method with 150000 independent iterations to calculate the mean estimation error covariance, which is shown in Fig. \ref{f3}. The stochastic event-triggered policies proposed in our work not only outperform the random offline scheduler, but also reduce the mean error covariance compared to the optimal offline periodic scheduler, especially when the communication rate is not sufficient to allow the persistent data transmissions transmit packets. Moreover, the closed-loop scheduler is better than the open-loop scheduler especially for $ \gamma\in [0.1,0.4] $ in this case.
	\begin{figure}[!t]
		\centering
		\includegraphics[width=0.9\columnwidth]{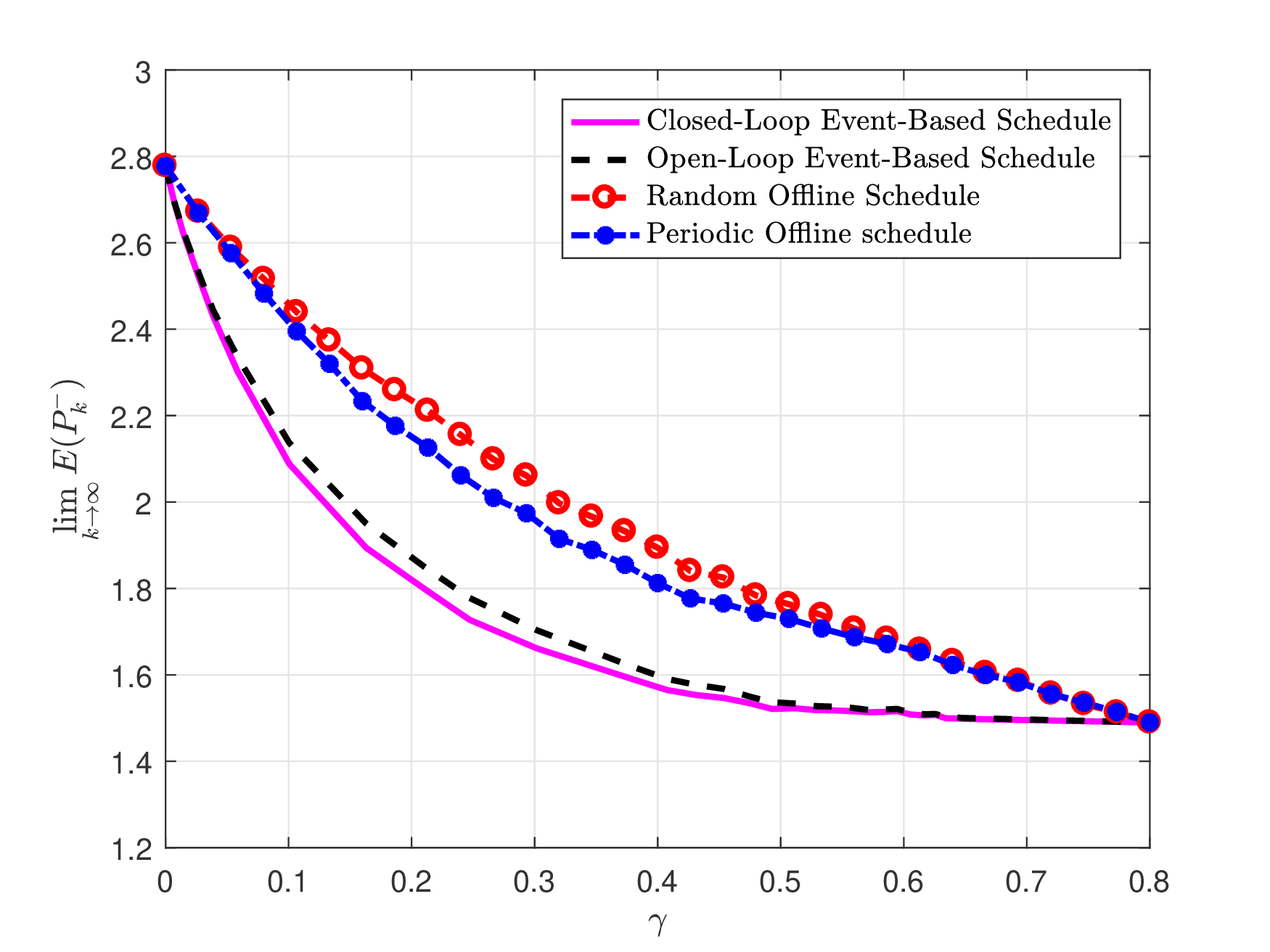}
		\caption{Empirical mean error covariance under three scheduling strategies versus effective communication rate.}
		\label{f3}
	\end{figure}

	\subsection{Performance Bounds}\label{c2}
	We consider the network availability rate $ \lambda=0.8 $. Fig \ref{f4} demonstrate the asymptotic bounds of mean error covariance in Lemma \ref{l2} for a stable system with parameters (same parameters as \cite{han2015stochastic} for comparison)
	\begin{align*}
	A=\left[ \begin{array}{cc}
	0.8& 1 \\ 
	0& 0.95
	\end{array}  \right] ,C=\left[ \begin{array}{cc}
	1& 1
	\end{array}  \right] ,Q=\left[ \begin{array}{cc}
	1& 0 \\ 
	0& 1
	\end{array}  \right],
	R=1 ,
	\end{align*} using an open-loop scheduler by 60000 simulation runs.
	\begin{figure}[!t]
		\centering
		\includegraphics[width=0.9\columnwidth]{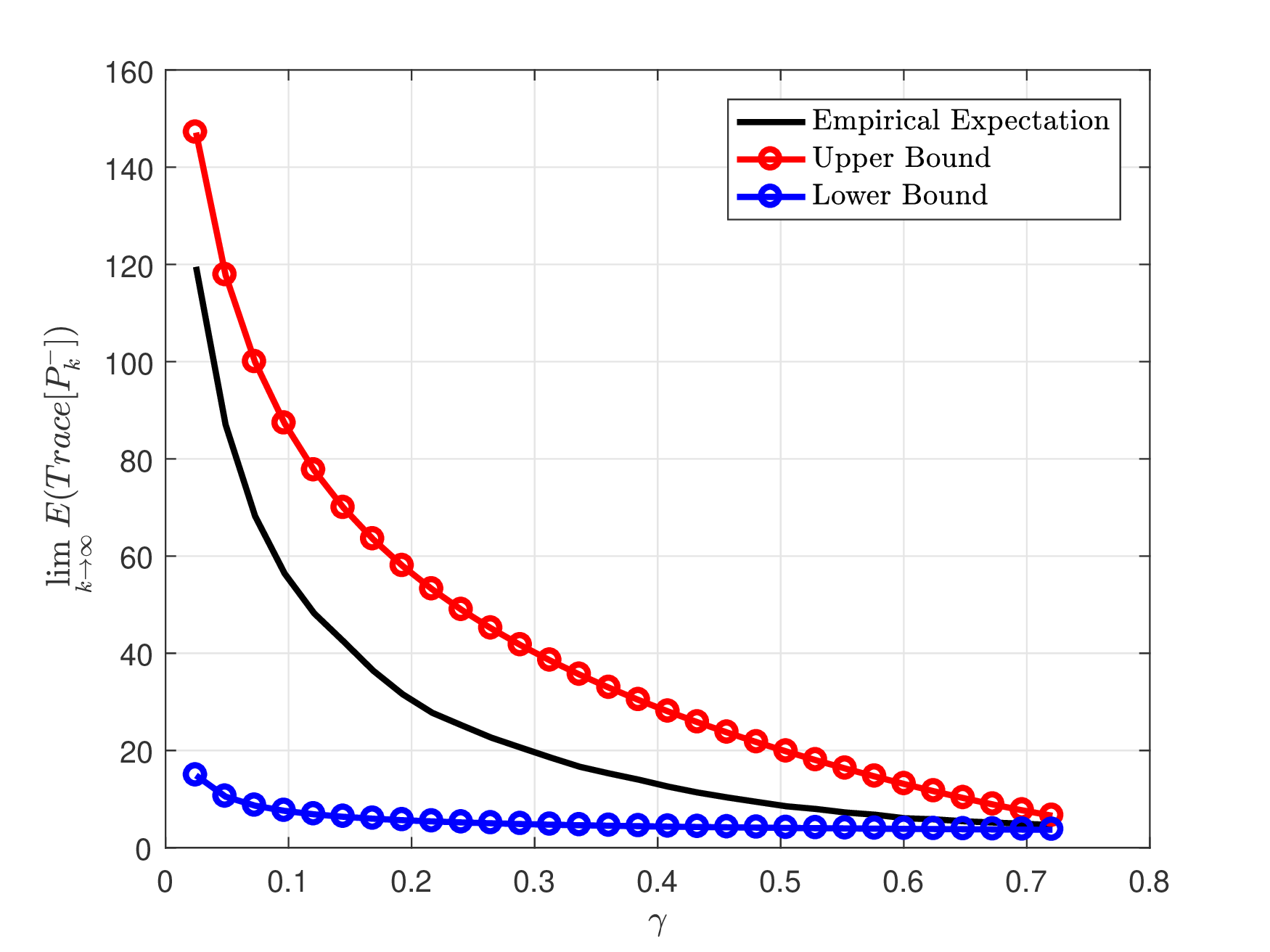}		
		\caption{Trace of the asymptotic upper bound $ \bar{X}_{ol} $ and the lower bound $ \underline{X}_{ol} $ of the open-loop scheduler versus empirical data.}
		\label{f4}
	\end{figure}
	We observe that when the communication rate $ \gamma $ is closer to $ \lambda $, the traces of the bounds for both cases are tighter. Similar results exist for an unstable system under the closed-loop scheduler.

	\subsection{Design of Event-triggered Parameter} \label{c3}	
	We assume $ \lambda=0.8 $. To compare this result with the schedulers proposed in~\cite{han2015stochastic}, we use the same system parameters 
	\begin{align*}
	&A=\left[ \begin{array}{cc}
	0.8& 1 \\ 
	0& 0.95
	\end{array}  \right] ,C=\left[ \begin{array}{cc}
	0.5& 0.3\\
	0 &1.4
	\end{array}  \right] ,\\
	&Q=\left[ \begin{array}{cc}
	1& 0 \\ 
	0& 1
	\end{array}  \right] ,R=\left[ \begin{array}{cc}
	1& 0 \\ 
	0& 1
	\end{array} \right] ,
	\end{align*} with the open-loop scheduler. Note that
	\begin{equation*}
	X_{0}=\left[ \begin{array}{cc}
	2.4353& 0.3976 \\ 
	0.3976& 1.3756
	\end{array} \right].
	\end{equation*}
	
	\begin{figure}[!t]
		\centering		
		\includegraphics[width=0.9\columnwidth]{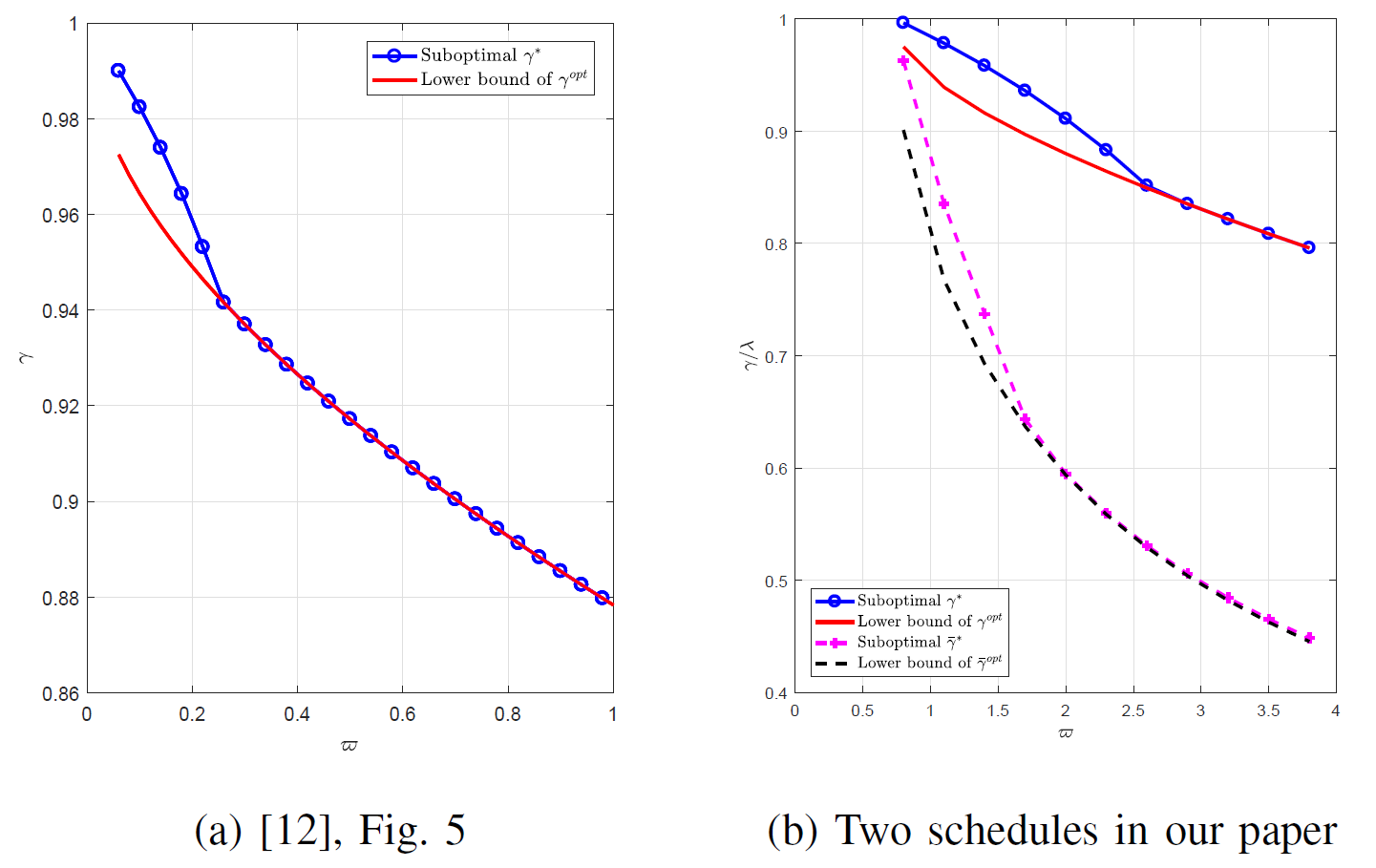}
		\caption{Comparison of the trade-off between the system performance quality and the communication rate.}
		\label{f6}			
	\end{figure}	
	The system quality constraint is
	$ 	M=X_{0}+\varpi I $,
	where $ \varpi $ is a positive real number. The suboptimal solution in Theorem \ref{t2} is obtained under different values of $ \varpi $, and it is shown in Fig. \ref{f6}(b) by the blue line. The same as in~\cite[Fig. 5]{han2015stochastic},  shown in Fig. \ref{f6}(a), the suboptimal solution equals the optimal solution when $ \varpi  $ is large, though the equivalent point of $ \varpi  $ is larger than that in~\cite{han2015stochastic}. Given the same quality constraint, the percentage of triggering the scheduling when the network is available is higher in our paper than that in \cite{han2015stochastic}. This is consistent with the fact that due to the induced uncertainty of the network access, more information is needed to guarantee the same estimation quality.
	
	Moreover, the suboptimal solution which follows Theorem \ref{t4} using a B\&B algorithm is shown by purple dashed line in Fig. \ref{f6}(b). We observe that to achieve the same estimation  quality, the upper bound of the communication rate using the closed-loop scheduler is much smaller than the communication rate using the open-loop scheduler. The suboptimal solution is also equivalent to the optimal solution when  $ \varpi  $ is large.This scenario has not been addressed in~\cite{han2015stochastic}.

	\subsection{Comparison between Different Access Probabilities $ \lambda $} \label{s7.4}
	In this subsection, we illustrate the scheduling performance by varying $ \lambda $. A scalar stable system with parameters $ A =0.8,C=1,Q=1$ and $R=1$ by an open-loop scheduler is considered. We adopt the Monte Carlo method using 50 independent sample paths with 3000 time steps each to calculate the mean error covariance. The results are shown in Fig. \ref{f7}. Two lower bounds are also plotted:
	\begin{enumerate}
		\item The lower bound (black filled dots) from Lemma \ref{l2} with $ \gamma=\lambda $, i.e., $ X_{0}$ \eqref{e15}.
		\item The lower bound (red hollow dots) in~\cite{sinopoli2004kalman}, i.e., $ X_{p}$  \eqref{eq10} .
	\end{enumerate} 
	
	The black dots in Fig. \ref{f7} are closer to the empirical results for fixed $ \lambda $ compared with the red dots when $ \lambda=0.4,0.6,0.8,1 $. If $ \lambda=0.2 $, the red dot is better than the black one. This coincides with the result that the lower bound derived in our paper is larger than the previous one in \cite{sinopoli2004kalman} especially when $ \lambda $ is large.
	\begin{figure}[!t]
		\centering
		\includegraphics[width=0.9\columnwidth]{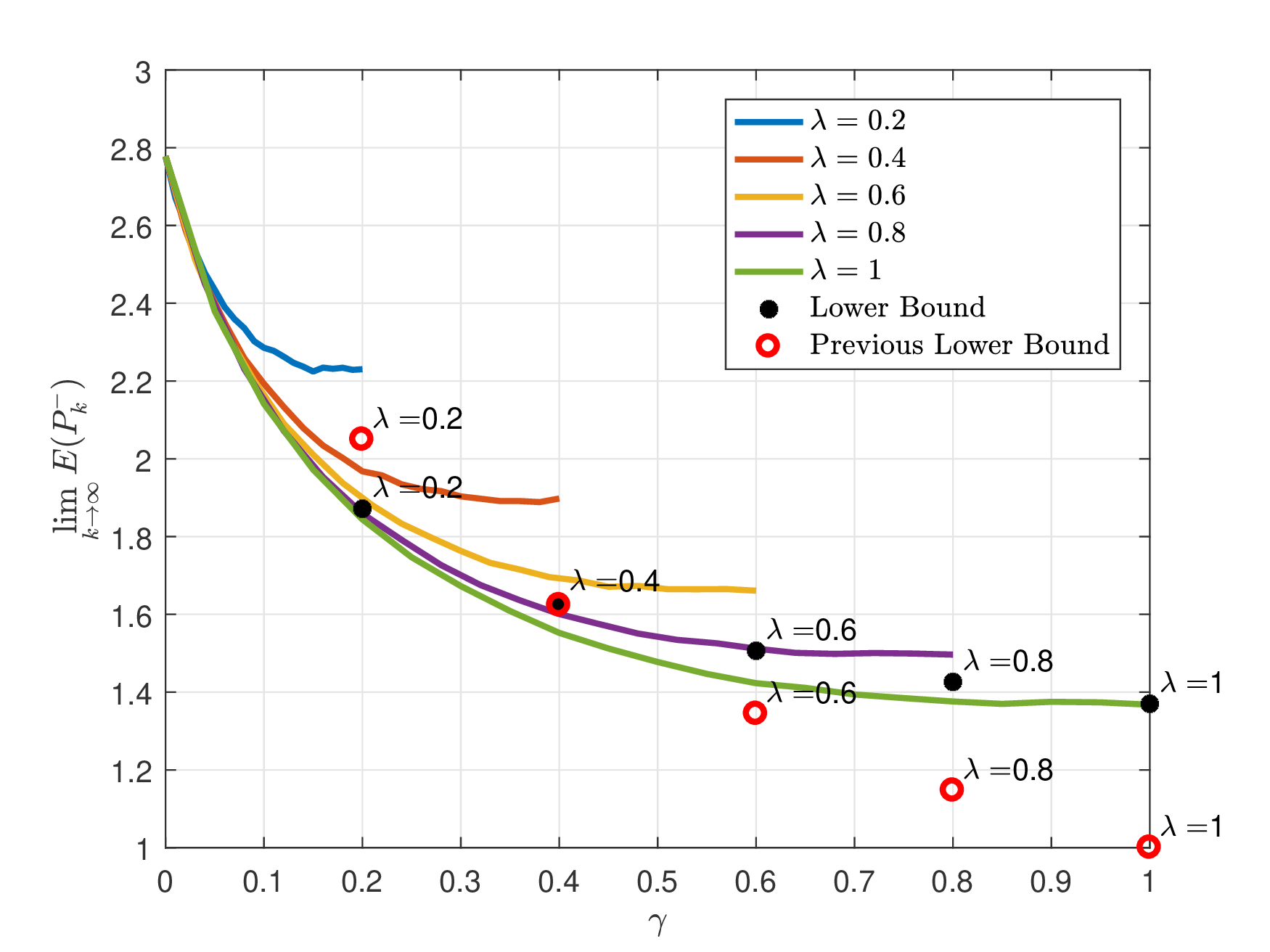}	
		\caption{Empirical mean error covariance under different network access probabilities and lower bounds comparision.}
		\label{f7}
	\end{figure}
	The unstable system under closed-loop scheduler has similar results.

	\section{conclusion}\label{sc}
	In this work, we developed stochastic event-triggered schedules for remote estimation in which the network access is uncertain. We started from the formulation of a discrete-time remote estimation process in the CRSN, and then analyzed the hidden information contained in the absence of data transmission. In order to achieve a better tradeoff between estimation performance and communication consumption, we proposed both open-loop and closed-loop schedules. Utilizing the hidden information, the MMSE estimators for both schedules were derived. The problem of minimizing the average communication rate while upholding a level of quality was studied. We proposed a suboptimal expression to design event parameter in the open-loop scenario by solving an SDP problem. Since the closed-form  of the communication rate cannot be obtained in the closed-loop scenario, a jointly biconvex problem is used to minimize the upper bound of the communication rate satisfying the quality constraint; the related global optimal boundary solution is obtained by B\&B algorithm.  Numerical examples were  provided to illustrate our results. Future work includes safety issues and multiple sensor schedulings in this system structure. It is also interesting to consider other network channel models, e.g., multi-state Markov chain.

	\appendix
	\subsection{ Propositions of function $ g_{\theta,W}  $} \label{a2}
	We first prove some useful properties of the matrix function $ 	g_{\theta,W}(X) $. 
	
	\begin{proposition}\label{p1}
		For all $ X_{1},X_{2}\in \mathbb{S}_{+}^{n} $, we have the following properties of $ g_{\theta,W} $:
		\begin{enumerate}		
			\item \emph{Monotonicity:} If $ X_{1}\geq X_{2} $, then $ g_{\theta,W}(X_{1})\geq g_{\theta,W}( X_{2})\geq Q$;
			\item \emph{Existence and uniqueness of a fixed point:} There exists a unique positive-definite $ X_{\theta,W}^{*} $ such that $ X_{*}=g_{\theta,W}(X_{*}) $;
			\item \emph{Limit property of the iterated function:} $ g_{\theta,W}^{k}(X)\rightarrow X_{\theta,W}^{*} $, for any $ X\in \mathbb{S}_{+}^{n} $ as $ k\rightarrow\infty $;
			\item \emph{Concavity:} For $ \forall\alpha\in[0,1] $, $ g_{\theta,W}(\alpha X_{1}+(1-\alpha)X_{2})\geq \alpha g_{\theta,W}(X_{1})+(1-\alpha)g_{\theta,W}(X_{2}) $. Therefore, by Jensen's inequality, one has $ \mathbb{E}(g_{\theta,W}(X))\leq g_{\theta,W}[\mathbb{E}(X)] $;	
			\item \emph{Monotonicity on $ W $:} For any $ W_{1}\geq W_{2} $, $  g_{\theta,W_{1}}(X)\geq g_{\theta,W_{2}}(X) $, $ X_{W_{1}}^{*}\geq X_{W_{2}}^{*} $.
		\end{enumerate}
	\end{proposition}
	\begin{proof}
		1)- 4) are proved in~\cite{sinopoli2004kalman} in detail.	
		
		5) For any $ W_{1}\geq W_{2} $, $ CXC^{\top}+W_{1}\geq CXC^{\top}+W_{2}\geq 0$, $ (CXC^{\top}+W_{1})^{-1}\leq( CXC^{\top}+W_{2})^{-1} $, $  \theta AXC^{\top}(CXC^{\top}+W_{1})^{-1}CXA^{\top}\leq \theta AXC^{\top}(CXC^{\top}+W_{2})^{-1}CXA^{\top}   $ and $  g_{\theta,W_{1}}(X)\geq g_{\theta,W_{2}}(X) $. The second equation holds because of the monotonicity property of the iterated function.
		
	\end{proof}
	
	\subsection{ Proof of Theorem \ref{t1} } \label{a1}
	Since the process $ x_{0} $ has a prior Gaussian distribution, i.e., $ x_{0}\sim\mathcal{N}(0,\Sigma) $, one can prove the MMSE estimate in a recursive way. Assume $ x_{k} $ has a prior Gaussian distribution as $ x_{k}\sim\mathcal{N}(x_{k}^{-},P_{k}^{-}) $. We need to prove the estimation update of $ x_{k} $ based on the new update $ \eta_{k} $, $ \varepsilon_{k} $ and $ \eta_{k}\varepsilon_{k}y_{k} $ are shown in \eqref{eq7}. Furthermore, the next estimation of $ x_{k+1} $ conditioned on the information set $ \mathcal{I}_{k} $ is shown in \eqref{q6}.
	
	Consider the following two cases:
	\begin{enumerate}
		\item When $ \eta_{k}=0 $,
		\begin{equation*}
		\begin{split}
		\hat{x}_{k}^{-}=\hat{x}_{k}^{-},
		P_{k}		=	P_{k}^{-}.
		\end{split}  					
		\end{equation*}
		\item When $ \eta_{k}=1 $,	from~\cite{han2015stochastic}'s Theorem 1,
		\begin{equation*}
		\begin{split} 	
		K_k'&= P_{k}^{-}C^{\top}(CP_{k}^{-}C^{\top} + R+(1-\varepsilon_{k})Y^{-1})^{-1},\\
		\hat{x}_{k}	& =(I-K_{k}'C)\hat{x}_{k}^{-}+\varepsilon_{k}K_k'y_{k}, 	\\
		P_{k}			&=	P_{k}^{-} -K_k'	C P_{k}^{-}.			 	
		\end{split}
		\end{equation*}		
	\end{enumerate}	
	Hence, from the above analysis, the above recursive equations are satisfied, where $ K_{k}=\eta_{k}K_{k}'$. This completes the measurement update proof.
	
	Then we consider the pdf of the time update. It is a Gaussian process  $ f(x_{k+1}|\mathcal{I}_{k})=f(Ax_{k}+w_{k}|\mathcal{I}_{k})=\mathcal{N}(A\hat{x}_{k},AP_{k}A^{\top}+Q) $, which is directly derived given that $ x_{k} $ and $ w_{k} $ are mutually independent Gaussian. This is the same as equations in \eqref{q6}. Thus, the proof is completed.

	\subsection{Proof of Lemma \ref{l2}} \label{a3} 
	
	We prove Lemma \ref{l2} by induction. For simplicity, denote $ U_{k}\triangleq g_{\lambda,R+Y^{-1}}^{k}(\Sigma) $.
	
	Clearly, $ \mathbb{E}[P_{0}^{-}]= U_{0}=\Sigma$. Assume $ \mathbb{E}[P^{-}_{k}]\leq U_{k} $. Then the statement is equal to proving that $  \mathbb{E}[P^{-}_{k+1}]\leq U_{k+1} $ From equations \eqref{eq7} and \eqref{q6}, one obtains
	\begin{align*}		
	\mathbb{E}[P^{-}_{k+1}]= &\mathbb{E}[g_{\lambda,R+(1-\varepsilon_{k})Y^{-1}}(P^{-}_{k})]
	\leq \mathbb{E}[g_{\lambda,R+Y^{-1}}(P^{-}_{k})\\
	\leq&[g_{\lambda,R+Y^{-1}}\mathbb{E}[(P^{-}_{k})]
	\leq g_{\lambda,R+Y^{-1}}(U_{k})=U_{k+1},
	\end{align*}
	where the first inequality holds from the fifth statement in Proposition \ref{p1}. The second inequality holds due to the concavity of function $ g_{\theta,W} $ and the last inequality holds recalling that $ g_{\theta,W} $ is a monotonically increasing function.
	
	From the above analysis, $ \mathbb{E}[P^{-}_{k}]\leq U_{k} $ for all $ k $ by induction. Moreover,  by Proposition \ref{p1}, $  U_{k}\rightarrow\bar{X}_{ol} $, as $ k\rightarrow\infty $, which implies that
	\begin{equation*}
	\limsup\limits_{k\rightarrow\infty}\mathbb{E}[P_{k}^{-}]\leq\bar{X}_{ol}.
	\end{equation*}
	
	On the other hand, to derive the lower bound, let us define 
	\begin{equation*}
	S_{k}\triangleq P_{k}^{-1},S_{k}^{-}\triangleq (P_{k}^{-})^{-1}.
	\end{equation*}
	
	There are three cases of the recursive function of $ P_{k} $ \eqref{eq7} as follows
	\begin{equation}\label{e1}
	\begin{split}
	P_{k}=\left\lbrace \begin{array}{l}
	P_{k}^{-}, \text{ if }\eta_{k}=0,\\ 
	P_{k}^{-}-P_{k}^{-}C^{\top}(R+CP_{k}^{-}C^{\top})^{-1}CP_{k}^{-}, \text{ if }\eta_{k}, \varepsilon_{k}=1,\\ 
	P_{k}^{-}-P_{k}^{-}C^{\top}(R+Y^{-1}+CP_{k}^{-}C^{\top})^{-1}CP_{k}^{-},  \text{ else}.
	\end{array}  \right.
	\end{split}	 
	\end{equation}
	Inverting both side of \eqref{e1}, we have
	\begin{equation}\label{e2}
	\begin{split}
	S_{k}=\left\lbrace \begin{array}{l}
	S_{k}^{-}, \text{ if }\eta_{k}=0,\\ 
	S_{k}^{-}+C^{\top}R^{-1}C, \text{ if }\eta_{k}=1, \varepsilon_{k}=1,\\ 
	S_{k}^{-}+C^{\top}(R+Y^{-1})^{-1}C, \text{ if }\eta_{k}=1, \varepsilon_{k}=0.
	\end{array}  \right.
	\end{split}	 
	\end{equation}
	Aggregating \eqref{e2}, one has 
	\begin{equation*}
	S_{k}=S_{k}^{-}+\eta_{k}(1-\varepsilon_{k})C^{\top}(R+Y^{-1})^{-1}C+\eta_{k}\varepsilon_{k}C^{\top}R^{-1}C.
	\end{equation*}	
	Taking the expectation of both sides, one has
	\begin{equation*}
	\begin{split}
	\mathbb{E}[\eta_{k}(1-\varepsilon_{k})]=\Pr(\eta_{k}=1,\varepsilon_{k}=0)=&\lambda-\gamma,		\\	
	\mathbb{E}[\eta_{k}\varepsilon_{k}]=\Pr(\eta_{k}=1,\varepsilon_{k}=1)	=&\gamma.
	\end{split}	
	\end{equation*}  Thus, we obtain
	\begin{equation*}
	\mathbb{E}[S_{k}]=\mathbb{E}[S_{k}^{-}]+C^{\top}R_{1}^{-1}C.
	\end{equation*}
	
	Meanwhile, the third equation in \eqref{q6} is the same as 
	\begin{equation*}
	S_{k+1}^{-}=(AS_{k}^{-1}A^{\top}+Q)^{-1},
	\end{equation*}
	from which $ S_{k+1}^{-} $ is concave w.r.t. $ S_{k} $~\cite{yang2013schedule}. By Jensen's inequality, the following inequality holds:
	\begin{equation*}
	\mathbb{E}[S_{k+1}^{-}]\leq (A\mathbb{E}[S_{k}]^{-1}A^{\top}+Q)^{-1}=\Lambda_{R_{1}}(\mathbb{E}[S_{k}^{-}]),
	\end{equation*}
	where $ \Lambda_{W}(X)\triangleq[A(X+C^{\top}W^{-1}C)^{-1}A^{\top}+Q]^{-1} $, $ W>0 $ for simplicity. 
	
	For any $ X_{1}\geq X_{2}\geq 0 $, the following equation holds:
	\begin{align*}
	&[\Lambda_{W}(X_{1})]^{-1}=g_{W}(X_{1}^{-1}) \leq g_{W}(X_{2}^{-1})= [\Lambda_{W}(X_{2})]^{-1},\\
	&\Leftrightarrow \Lambda_{W}(X_{1})\geq\Lambda_{W}(X_{2}).
	\end{align*}
	Hence, the monotonicity of $ \Lambda_{W} $ is proven.
	
	Furthermore, $ [\Lambda_{W}(X^{-1})]^{-1}=g_{W}(X) $ holds by applying the matrix inversion lemma and the result directly follows
	\begin{equation*}
	[\Lambda_{W}^{k}(X^{-1})]^{-1}=g_{W}^{k}(X). 
	\end{equation*}
	Based on the monotonicity of $ \Lambda_{R_{1}} $ and $ S_{0}^{-}=\Sigma^{-1} $, we obtain
	\begin{equation*}
	\mathbb{E}[S_{k}^{-}]\leq \Lambda_{R_{1}}(\mathbb{E}[S_{k-1}^{-}]) \leq\cdots\leq \Lambda_{R_{1}}^{k}(\Sigma^{-1}).
	\end{equation*}
	Since $ f(X)=X^{-1},X\geq 0 $ is convex w.r.t $ X $, and by Jensen's inequality, one has
	\begin{equation*}
	\begin{split}
	\mathbb{E}[P_{k}^{-}]=\mathbb{E}[(S_{k}^{-})^{-1}]\geq \mathbb{E}[(S_{k}^{-})]^{-1}
	\geq [\Lambda_{R_{1}}^{k}(\Sigma^{-1})]^{-1}=g_{R_{1}}^{k}(\Sigma).
	\end{split}		
	\end{equation*}
	
	Denote $ D_{k}\triangleq g_{R_{1}}^{k}(\Sigma)$. From the above analysis, one has $ \mathbb{E}[P^{-}_{k}] \geq D_{k} $ for all $ k $. By Proposition \ref{p1}, $  D_{k}\rightarrow\underline{X}_{ol} $, as $ k\rightarrow\infty $, which implies that
	\begin{equation*}
	\liminf\limits_{k\rightarrow\infty}\mathbb{E}[P_{k}^{-}]\geq\underline{X}_{ol}.
	\end{equation*}
	From the fifth statement in Proposition \ref{p1}, $ \underline{X}_{ol}\geq X_{0} $ always holds as $ R_{1}\geq\dfrac{R}{\lambda} $.
	The proof is done.
	
	\subsection{ Proof of Theorem \ref{t2}} \label{a4}
	The proof of Theorem \ref{t2} follows the following two steps. 		
	First, we prove an equivalent set of constraints to replace the implicit constraint $ \bar{X}_{ol} \leq M $. Second, the set of constraints are transformed  to an SDP constraint.
	
	Firstly, the following two statements are equivalent:
	\begin{enumerate}
		\item $ \bar{X}_{ol} \leq M $,
		\item There exists $ 0<X\leq M $ such that $ g_{\lambda,R+Y^{-1}}(X)\leq X $.
	\end{enumerate} 
	$ ``1)\Rightarrow2) ''$: It is obvious that the second statement can be obtained from the first, i.e., $ \bar{X}_{ol} $ is a feasible solution to $ X $.\\	
	$ ``2)\Rightarrow1) ''$: Recall that $ g_{\lambda,W}(X) $ is a monotonically increasing function in $ X $ from Proposition \ref{p1}.  We have
	\begin{align*}
	M\geq X\geq g_{\lambda,R+Y^{-1}}(X)\geq g_{\lambda,R+Y^{-1}}^{2}(X)\\
	\geq \ldots \geq \lim\limits_{k\rightarrow\infty}g_{\lambda,R+Y^{-1}}^{k}(X)=\bar{X}_{ol}.
	\end{align*} 
	Then the first statement is obtained from the second. Thus, these two statements are equivalent. 
	
	The constraints of Problem \ref{q3} are rewritten as follows:
	\begin{align} \label{eq58}
	Y\geq0,
	0<X\leq M,
	g_{\lambda,R+Y^{-1}}(X)\leq X.
	\end{align} 
	
	Secondly, the main difficulty is to transform the last inequality into an equivalent SDP constraint. Since the last inequality cannot be changed to linear form based on $ X $, we transform it to linear form based on the inverse of X, i.e., $ S $. To maintain the parameter utility, the second inequality should also be changed to the linear form based on $ S $.	
	Taking the inverse of both sides of the second inequality in \eqref{eq58}, we obtain	
	$ S\geq M^{-1} $.
	It is straightforward to see that
	\begin{equation*}
	S\geq M^{-1}\Leftrightarrow\left[ \begin{array}{cc}
	S& I  \\ 
	I&  M
	\end{array} \right] \geq 0,
	\end{equation*}
	by Schur complement since $ S=X^{-1}>0 $.
	
	The left-hand part of the problem is to transform the third inequality in \eqref{eq58} to an SDP form in $ S $. By rearranging the term, one has
	\begin{equation}\label{eq22}
	\begin{split}
	&X-(1-\lambda)AXA^{\top}-Q\\
	&-\lambda(AXA^{\top}- AXC^{\top}(CXC^{\top}+R+Y^{-1})^{-1}CXA^{\top})\\
	&=X-(1-\lambda)AXA^{\top}-Q\\
	&-\lambda A(S+C^{\top}(R+Y^{-1})^{-1}C)^{-1}A^{\top} \geq0, 
	\end{split} 		
	\end{equation} 
	where the equality follows the matrix inversion lemma.
	
	Since $ S>0,Y\geq0,R>0 $, the following equation holds
	\begin{equation}\label{eq23}
	S+C^{\top}(R+Y^{-1})^{-1}C>0,
	\end{equation}		
	then by applying the Schur complement to \eqref{eq22}, the third inequality in \eqref{eq58} is the same as 
	\begin{equation}\label{eq44}
	\left[ \begin{array}{cc}
	X-(1-\lambda)AXA^{\top}-Q& \sqrt{\lambda}A \\ 
	\sqrt{\lambda}A^{\top} & S+C^{\top}(R+Y^{-1})^{-1}C
	\end{array} \right]\geq 0 .	
	\end{equation}

	We obtain $ X-Q-(1-\lambda)AXA^{\top}\geq0 $ from \eqref{eq44}. Meanwhile, as $ S=X^{-1}>0 $,  the following inequality holds: 
	\begin{equation}\label{eq45}
	\left[ \begin{array}{cc}
	X-Q& \sqrt{1-\lambda}A\\ 
	\sqrt{1-\lambda}A^{\top}&S
	\end{array} \right]\geq 0. 
	\end{equation}
	
	Given that $ X-Q\geq0 $ from \eqref{eq45} and $ X>0 $, it is straightforward to see that
	\begin{equation}\label{eq46}
	\left[ \begin{array}{cc}
	X& I\\ 
	I&	Q^{-1}
	\end{array} \right]\geq 0. 
	\end{equation}
	
	Combining \eqref{eq44}, \eqref{eq45} and \eqref{eq46}, one has	
	\begin{equation}
	\begin{split}
	&\Theta\triangleq\\
	&\left[ \begin{array}{cccc}
	X& \sqrt{\lambda}A &\sqrt{1-\lambda}A&I\\ 
	\sqrt{\lambda}A^{\top} & S+C^{\top}(R+Y^{-1})^{-1}C&0&0\\
	\sqrt{1-\lambda}A^{\top}&0&S&0\\
	I&0&0&Q^{-1}
	\end{array} \right]\\
	&\geq 0 .
	\end{split} 		
	\end{equation}
	
	This is equivalent to 
	\begin{equation}
	\begin{split}
	&\Gamma\triangleq\left[ \begin{array}{cccc}
	S& 0 & 0 &0\\ 
	0&  I& 0 & 0\\ 
	0& 0 & I & 0\\
	0& 0 & 0 & I
	\end{array} \right] \Theta\left[ \begin{array}{cccc}
	S& 0 & 0 &0\\ 
	0&  I& 0 & 0\\ 
	0& 0 & I & 0\\
	0& 0 & 0 & I
	\end{array} \right]\geq 0 \Leftrightarrow\\
	&\left[ \begin{array}{cccc}
	S& \sqrt{\lambda}SA &\sqrt{1-\lambda}SA&S\\ 
	\sqrt{\lambda}A^{\top}S & \Gamma_{22}&0&0\\
	\sqrt{1-\lambda}A^{\top}S&0&S&0\\
	S&0&0&Q^{-1}
	\end{array} \right]		\geq 0 ,
	\end{split}
	\end{equation}
	where $ \Gamma_{22}\triangleq S+C^{\top}(R+Y^{-1})^{-1}C $.  
	
	Since $ \Gamma_{22} $ is not linear in $ Y $,  we expand $ (R+Y^{-1})^{-1} $ by using the matrix inversion lemma, where
	\begin{equation*}
	(R+Y^{-1})^{-1}=R^{-1}-R^{-1}(Y+R^{-1})^{-1}R^{-1}.
	\end{equation*}
	Then one has
	\begin{equation}\label{eq76}
	\begin{split}
	\Gamma=&\left[ \begin{array}{cccc}
	S& \sqrt{\lambda}SA &\sqrt{1-\lambda}SA&S\\ 
	\sqrt{\lambda}A^{\top}S & S+C^{\top}R^{-1}C&0&0\\
	\sqrt{1-\lambda}A^{\top}S&0&S&0\\
	S&0&0&Q^{-1}
	\end{array} \right]\\
	&-\left[ \begin{array}{c}
	0\\ 
	C^{\top}R^{-1}\\ 
	0\\ 
	0
	\end{array} \right] (Y+R^{-1})^{-1}\left[ \begin{array}{c}
	0\\ 
	C^{\top}R^{-1}\\ 
	0\\ 
	0
	\end{array} \right]^{T}\geq0. 
	\end{split}		
	\end{equation}
	
	As $ Y\geq0,R>0 $, $ (Y+R^{-1})^{-1}>0 $ holds. The above inequality \eqref{eq76} can also be  viewed as a Schur complement, where
	\begin{equation*}
	\begin{split}
	\tilde{A}&\triangleq\left[ \begin{array}{cccc}
	S& \sqrt{\lambda}SA &\sqrt{1-\lambda}SA&S\\ 
	\sqrt{\lambda}A^{\top}S & S+C^{\top}R^{-1}C&0&0\\
	\sqrt{1-\lambda}A^{\top}S&0&S&0\\
	S&0&0&Q^{-1}
	\end{array} \right],\\
	\tilde{B}^{T}&\triangleq\left[ \begin{array}{cccc}
	0& 		R^{-1}C &		0&		0
	\end{array}\right], \tilde{C}\triangleq Y+R^{-1}>0.
	\end{split}	
	\end{equation*} 		
	Given that $ \tilde{C}>0 $, then $ \Gamma= \tilde{A}-\tilde{B}\tilde{C}^{-1}\tilde{B}^{T}\geq0 $ if and only if $ \Psi(S,Y)=\left[ \begin{array}{cc}
	\tilde{A}&  \tilde{B}\\ 
	\tilde{B}^{T}& \tilde{C}
	\end{array}\right] \geq 0  $ . The proof is done.
	
	\subsection{Proof of Theorem \ref{t3}}\label{at3}
	From Theorem \ref{t1}, one has
	\begin{equation*}
	\begin{split}
	f(z_{k}|\mathcal{I}_{k-1})=&f(y_{k}-\hat{y}_{k}^{-}|\mathcal{I}_{k-1})=f(Ce_{k}^{-}+v_{k}|\mathcal{I}_{k-1})\\
	=&\mathcal{N}(0,CP_{k}^{-}C^{\top}+R),
	\end{split}		
	\end{equation*}
	where the second equation holds as $ \hat{y}_{k}^{-}=C\hat{x}_{k}^{-}$ from Theorem \ref{t1}; the last equation holds as $ \mathbb{E}[e_{k}^{-}|\mathcal{I}_{k-1}]=0 $, and $ e_{k}^{-} $, $ v_{k} $ are mutually independent Gaussian variables.
	
	For the measurement update, performing a similar analysis, we can obtain \eqref{e33}. Note that, substituting $ y_{k} $ by $ z_{k} $ in \eqref{eq7},		
	$ \hat{x}_{k} =\hat{x}_{k}^{-}+\varepsilon_{k}K_kz_{k}-K_{k}\mathbb{E}[z_{k}|\mathcal{I}_{k-1}]=\hat{x}_{k}^{-}+\varepsilon_{k}K_kz_{k}  $,	
	which is consistent with \eqref{e33}. We omit the remainder proof as it is straight forward.
	
	\subsection{Proof of Lemma \ref{l4}}\label{a5}
	We have
	\begin{equation*}
	\begin{split}
	\gamma=&\mathbb{E}[\Pr(\eta_{k}\varepsilon_{k}=1|\mathcal{I}_{k-1})]\\
	=&\mathbb{E}[\Pr(\eta_{k}=1)\Pr(\zeta_{k}\leq exp(-\frac{1}{2}z_{k}^{T}Zz_{k})|\mathcal{I}_{k-1})] \\
	=&\lambda\mathbb{E}\left[ 1-\left( \text{det}\left( I+(CP_{k}^{-}C^{\top}+R)Z\right) \right)^{-\frac{1}{2}}\right] .
	\end{split}		
	\end{equation*}

	To prove the second inequality, it suffices to prove the concavity of \eqref{eq70}. By Jensen's inequality, it suffices to prove the convexity of function $ f(CXC^{\top}Z+RZ+I)\triangleq\left( \text{det}\left( I+(CP_{k}^{-}C^{\top}+R)Z\right) \right)^{-\frac{1}{2}} $.
	
	The convexity holds for a composition with affine functions; therefore, it is equivalent to prove that
	\begin{equation*}
	f(X)=[det(X)]^{-\frac{1}{2}}, \text{ for } X\geq I,
	\end{equation*}
	is convex.
	
	Define $ b:(0,\infty)\rightarrow \mathbb{R} $ and	$ b(s)\triangleq(s)^{-\frac{1}{2}} $.
	Lehmich el. al.~\cite{lehmich2014convexity} states the convexity of $ f $ on the set $ X\in \mathbb{S}_{++}^{n}$ is equivalent to
	\begin{equation*}
	nsb''(s)+(n-1)b'(s)\geq0 \text{, and } b'(s)\leq0 \text{, for all  } s>0.
	\end{equation*}
	Since 	$ 	b'(s)=-\frac{1}{2}(s)^{-\frac{3}{2}}\leq0 $		and 
	\begin{equation*}
	\begin{split}
	nsb''(s)+(n-1)b'(s)=&ns\frac{3}{4}(s)^{-\frac{5}{2}}-(n-1)\frac{1}{2}(s)^{-\frac{3}{2}}\\
	=&\frac{3}{4}(s)^{-\frac{3}{2}}+(n-1)\frac{1}{4}(s)^{-\frac{3}{2}}\geq0,
	\end{split}		  
	\end{equation*} 
	the proof is completed.
	
	\subsection{Proof of Theorem \ref{t4}}\label{a6}
	An equivalent statement of the constraints of Problem \ref{p4} is as follows.	
	The following two statements are equivalent:
	\begin{enumerate}
		\item $\bar{X}_{cl}\leq M ,C\bar{X}_{cl} C^{\top}+R\leq T$,
		\item There exists $  S^{-1}=X\geq M $ such that
		\begin{equation*}
		\begin{split}
		\Psi(S,Z) \geq 0, CX C^{\top}+R\leq T.
		\end{split}			
		\end{equation*} 
	\end{enumerate}
	
	$ 1)\Rightarrow2) $: Let $ X $ be equal to $ \bar{X}_{cl} $; it is obvious that $ \bar{X}_{cl} $ is a feasible matrix satisfying the second statement.
	
	$ 2)\Rightarrow1) $: From the similar proof in Appendix \ref{a4}, one has 
	\begin{align*}
	g_{\lambda,R+Z^{-1}}(X)\leq X,\bar{X}_{cl}\leq X\leq M;
	\end{align*}
	therefore, $  C\bar{X}_{cl} C^{\top}+R\leq CX C^{\top}+R\leq T$. 
	
	On the other hand, it is well known that replacing $ S^{-1}=X $ by $ S^{-1}\leq X $ does not affect the solution to the optimization problem \eqref{e44} since $ {\rm tr}\left((CX C^{\top}+R)Z\right)\geq{\rm tr}\left((CS^{-1} C^{\top}+R)Z\right) $. Therefore, $ S^{-1}=X $ is satisfied for at least one optimal solution to the optimization problem, which completes the proof.
	
	\subsection{Proof of Lemma \ref{ll7}} \label{a7}
	By Lemma \ref{l5} and Lemma \ref{l7}, the feasibility condition for the problem is that $ X_{0}\leq\bar{X}_{cl} $. Moreover, from the proof of Theorem \ref{t4}, one has $ X_{0}\leq\bar{X}_{cl}\leq X\leq M $. It is sufficient to obtain that $ l_{ii}\leq x_{ii}\leq m_{ii} $. Furthermore, since every principal sub-matrix is positive definite for a positive semidefinite matrix, we have $ |x_{ij}|\leq(x_{ii}x_{jj})^{\frac{1}{2}}\leq (m_{ii}m_{jj})^{\frac{1}{2}}$.
	
	As the objective function in equation \eqref{e44} satisfies
	\begin{align*}
	&\min\limits_{Z,X,S}\,  {\rm tr}\left((CX C^{\top}+R)Z\right)\leq\min\limits_{Z,X,S}\,  {\rm tr}(Z){\rm tr}(CX C^{\top}+R)\\
	&\leq {\rm tr}(Z)\min\limits_{X,S}\,{\rm tr}(CX C^{\top}+R) \leq z^{*} \min\limits_{X,S}\,{\rm tr}(CX C^{\top}+R),
	\end{align*}
	the following equation holds
	\begin{equation}\label{eq39}
	z^{*}\geq\dfrac{\min\limits_{Z,X,S}\,  {\rm tr}\left((CX C^{\top}+R)Z\right)}{\min\limits_{X,S}\,{\rm tr}(CX C^{\top}+R)}, {\rm s.t.} (X,Z,S)\in\mathcal{S},
	\end{equation}
	where $ z^{*}\geq {\rm tr}(Z)$ is the upper bound of $ {\rm tr}(Z) $ on~$ \mathcal{S}$. As the molecule  $ \min\limits_{Z,X,S}\,  {\rm tr}\left((CX C^{\top}+R)Z\right)\leq \min\limits_{Z,S}\,{\rm tr}\left((CM C^{\top}+R)Z\right) $ and the denominator  $ \min\limits_{X,S}\,{\rm tr}(CX C^{\top}+R)\geq {\rm tr}(CX_{0} C^{\top}+R)$, let $ z^{*}$ be the solution to \eqref{e38}, and we can prove that $ z^{*} $ satisfies \eqref{eq39} from the above analysis. Therefore, one has $ 0\leq z_{ii}\leq {\rm tr}(Z)\leq z^{*}$. Moreover, $  |z_{ij}|\leq(z_{ii}z_{jj})^{\frac{1}{2}}\leq z^{*}$. This completes the proof.

\end{document}